\documentclass[11pt]{article}

\usepackage{tikz}

\usepackage{amsmath,amsfonts,amsthm,amssymb,color}

\newtheorem{theorem}{Theorem}[section]

\newtheorem{lemma}[theorem]{Lemma}
\newtheorem{definition}[theorem]{Definition}

\newtheorem{case}{Case}

\theoremstyle{remark}

\usepackage{latexsym,bbm,xspace,graphicx,float}
\definecolor{newblue}{rgb}{0.19, 0.55, 0.91}
\usepackage[colorlinks,citecolor=newblue,bookmarks=true,pagebackref=true, urlcolor=blue, linkcolor=blue, linktoc=page]{hyperref}
\usepackage[nameinlink]{cleveref}
\usepackage[letterpaper,margin=1in]{geometry}

\usepackage{cleveref}

\usepackage{url}
\usepackage{mathtools}
\usepackage{graphicx}
\usepackage{multirow}
\usepackage[ruled,linesnumbered]{algorithm2e}

\newcommand{\eps}{\varepsilon}

\usepackage{mdframed}
\usepackage{color, colortbl}

\usepackage{enumitem}
\usepackage{thmtools}
\usepackage{thm-restate}

\usepackage{amsmath,amsfonts,bm}

\def\1{\bm{1}}

\DeclareMathAlphabet{\mathsfit}{\encodingdefault}{\sfdefault}{m}{sl}
\SetMathAlphabet{\mathsfit}{bold}{\encodingdefault}{\sfdefault}{bx}{n}

\newcommand{\calA}{\mathcal{A}}
\newcommand{\calB}{\mathcal{B}}

\newcommand{\R}{\mathbb{R}}

\renewcommand{\tilde}{\widetilde}
\renewcommand{\bar}{\overline}

\newcommand{\norm}[1]{\left\|#1\right\|}

\providecommand{\expect}[2]{\ensuremath{\ifthenelse{\equal{#1}{}}{\mathbb{E}}{\mathbb{E}_{#1}}\!\left[#2\right]}\xspace}
\providecommand{\prob}[2]{\ensuremath{\ifthenelse{\equal{#1}{}}{\Pr}{\Pr_{#1}}\!\left[#2\right]}\xspace}

\newcommand{\inner}[1]{\langle #1\rangle}

\usepackage{bm}

\DeclareMathOperator{\poly}{poly}

\newcommand{\sk}{\mathrm{sk}}

\newcommand{\card}[1]{\left \vert #1 \right \vert}

\begin{document}

\title{Tight Lower Bounds for Directed Cut Sparsification and Distributed Min-Cut}

\author{Yu Cheng\footnote{Brown University. \texttt{yu\_cheng@brown.edu}}
        \hspace{5.5em}
        Max Li\footnote{Carnegie Mellon University. \texttt{mlli@andrew.cmu.edu}}
        \hspace{5.5em}
        Honghao Lin\footnote{Carnegie Mellon University. \texttt{honghaol@andrew.cmu.edu}}
        \\[0.5em]
        Zi-Yi Tai\footnote{Carnegie Mellon University. \texttt{ztai@andrew.cmu.edu}}
        \hspace{3em}
        David P. Woodruff\footnote{Carnegie Mellon University. \texttt{dwoodruf@andrew.cmu.edu}}
        \hspace{3em}
        Jason Zhang\footnote{Carnegie Mellon University. \texttt{jasonz3@andrew.cmu.edu}}
       }
\date{\vspace{-5ex}}

\maketitle

\newcolumntype{L}[1]{>{\raggedright\let\newline\\\arraybackslash\hspace{0pt}}m{#1}}
\newcolumntype{C}[1]{>{\centering\let\newline\\\arraybackslash\hspace{0pt}}m{#1}}
\newcolumntype{R}[1]{>{\raggedleft\let\newline\\\arraybackslash\hspace{0pt}}m{#1}}



\maketitle


\begin{abstract}
In this paper, we consider two fundamental cut approximation problems on large graphs.
We prove new lower bounds for both problems that are optimal up to logarithmic factors.

The first problem is to approximate cuts in balanced directed graphs.
In this problem, the goal is to build a data structure that $(1 \pm \eps)$-approximates cut values in graphs with $n$ vertices.
For arbitrary directed graphs, such a data structure requires $\Omega(n^2)$ bits even for constant $\eps$. 
To circumvent this, recent works study $\beta$-balanced graphs, meaning that for every directed cut, the total weight of edges in one direction is at most $\beta$ times that in the other direction.
We consider two models: the {\em for-each} model, where the goal is to approximate each cut with constant probability, and the {\em for-all} model, where all cuts must be preserved simultaneously.
We improve the previous $\Omega(n \sqrt{\beta/\eps})$ lower bound to $\tilde{\Omega}(n \sqrt{\beta}/\eps)$ in the for-each model, and we improve the previous $\Omega(n \beta/\eps)$ lower bound to $\Omega(n \beta/\eps^2)$ in the for-all model.~\footnote{In this paper, we use $\tilde O(\cdot)$ and $\tilde \Omega(\cdot)$ to hide logarithmic factors in its parameters.} This resolves the main open questions of (Cen et al., ICALP, 2021).

The second problem is to approximate the global minimum cut in a local query model, where we can only access the graph via degree, edge, and adjacency queries. We improve the previous $\Omega\bigl(\frac{m}{k}\bigr)$ query complexity lower bound to $\Omega\bigl(\min\{m, \frac{m}{\eps^2 k}\}\bigr)$ for this problem, where $m$ is the number of edges, $k$ is the size of the minimum cut, and we seek a $(1+\eps)$-approximation.
In addition, we show that existing upper bounds with slight modifications match our lower bound up to logarithmic factors.
\end{abstract}

\maketitle

\section{Introduction}
The notion of cut sparsifiers has been extremely influential.
It was introduced by Bencz{\'{u}}r and Karger~\cite{BK96} and it is the following:
Given a graph $G = (V, E, w)$ with $n = |V|$ vertices, $m = |E|$ edges, edge weights $w_e \ge 0$, and a desired error parameter $\eps > 0$, a $(1\pm \eps)$ {\em cut sparsifier} of $G$ is a subgraph $H$ on the same vertex set $V$ with (possibly) different edge weights, such that $H$ approximates the value of every cut in $G$ within a factor of $(1\pm\eps)$. Bencz{\'{u}}r and Karger~\cite{BK96} showed that every undirected graph has a $(1\pm\eps)$ cut sparsifier with only $O(n \log n/\eps^2)$ edges. This was later extended to the stronger notion of spectral sparsifiers~\cite{ST11} and the number of edges was improved to $O(n/\eps^2)$~\cite{BSS12}; see also related work with different bounds for both cut and spectral sparsifiers \cite{FHHP11,KP12,ST04,SS11,LeeS17,CKST19}.

In the database community, a key result is the work of \cite{AGM12}, which shows how to construct a sparsifer using $\tilde O(n/\eps^2)$ linear measurements to $(1+\eps)$-approximate all cut values. Sketching massive graphs arises in various applications where there are entities and relationships between them, such as webpages and hyperlinks, people and friendships, and IP addresses and data flows. As large graph databases are often distributed or stored on external memory, sketching algorithms are useful for reducing communication and memory usage in distributed and streaming models. We refer the readers to \cite{mcgregor2014graph} for a survey of graph stream algorithms in the database community. 

For very small values of $\eps$, the $1/\eps^2$ dependence in known cut sparsifiers may be prohibitive. Motivated by this, the work of \cite{ACK+15} relaxed the cut sparsification problem to outputting a data structure $D$, such that for any fixed cut $S \subset V$, the value $D(S)$ is within a $(1 \pm \eps)$ factor of the cut
value of $S$ in $G$ with probability at least $2/3$. Notice the order of quantifiers --- the data structure only needs to preserve the value of any fixed cut (chosen independently of its randomness) with high constant probability. This is referred to as the {\em for-each} model, and the data structure is called a {\em for-each cut sketch}. Surprisingly, \cite{ACK+15} showed that every undirected graph has a $(1\pm\eps)$ for-each cut sketch of size $\tilde{O}(n/\eps)$ bits, reducing the dependence on $\eps$ to linear.
They also showed an $\Omega(n/\eps)$ bits lower bound in the for-each model.
The improved dependence on $\eps$ is indeed coming from relaxing the original sparsification problem to the for-each model: \cite{ACK+15} proved an $\Omega(n/\eps^2)$ bit lower bound on any data structure that preserves all cuts simultaneously, which is referred to as the {\em for-all} model. This lower bound in the for-all model was strengthened to $\Omega(n \log n/\eps^2)$ bits in \cite{CKST19}. 

While the above results provide a fairly complete picture for undirected graphs, a natural question is whether similar improvements are possible for directed graphs.
This is the main question posed by \cite{CCP+21}. For directed graphs, even in the for-each model, there is an $\Omega(n^2)$ lower bound without any assumptions on the graph.
Motivated by this, \cite{EMPS16, IT18, CCP+21} introduced the notion of {\em $\beta$-balanced} directed graphs, meaning that for every directed cut $(S, V \setminus S)$, the total weight of edges from $S$ to $V \setminus S$ is at most $\beta$ times that from $V \setminus S$ to $S$.
The notion of $\beta$-balanced graphs turned out to be very useful for directed graphs, as \cite{IT18,CCP+21} showed an $\tilde{O}(n \sqrt{\beta}/\eps)$ upper bound in the for-each model, and an $\tilde{O}(n \beta/\eps^2)$ upper bound in the for-all model, thus giving non-trivial bounds for both problems for small values of $\beta$.
The work of \cite{CCP+21} also proved lower bounds:
they showed an $\Omega(n \sqrt{\beta/\eps})$ lower bound in the for-each model, and an $\Omega(n \beta/\eps )$ lower bound in the for-all model.
While their lower bounds are tight for constant $\eps$, there is a quadratic gap for both models in terms of the dependence on $\eps$. The main open question of \cite{CCP+21} is to determine the optimal dependence on $\eps$, which we resolve in this work.

Recent work further explored {\em spectral sketches}, faster computation of sketches, and sparsification of Eulerian graphs ($\beta$-balanced graphs with $\beta = 1$)~\cite{ACK+15,JambulapatiS18,CohenKKPPRS18,ChuGPSSW23,SaranurakW19}.
In this paper, we focus on the space complexity of cut sketches for general values of $\beta$.

As observed in \cite{ACK+15}, one of the main ways to use for-each cut sketches is to solve the distributed minimum cut problem.
This is the problem of computing a $(1+\eps)$-approximate global minimum cut of a graph whose edges are distributed across multiple servers.
One can ask each server to compute a $(1 \pm 0.2)$ for-all cut sketch and a $(1 \pm \eps)$ for-each cut sketch.
This allows one to find all $O(1)$-approximate minimum cuts, and because there are at most $n^{O(C)}$ cuts with value within a factor of $C$ of the minimum cut, one can query all these $\poly(n)$ cuts using the more accurate for-each cut sketches, resulting in an optimal linear in $1/\eps$ dependence in the communication.

Motivated by this connection to distributed minimum cut estimation, we also consider the problem of directly approximating the minimum cut in a {\em local query} model, which was introduced in \cite{RSW18} and studied for minimum cut in \cite{eden2017lower,globalmincut}. The model is defined as follows. 

Let $G(V, E)$ be an unweighted and undirected graph, where the vertex set $V$ is known but the edge set $E$ is unknown. In the {\em local query} model, we have access to an oracle that can answer the following three types of local queries:
\begin{enumerate}[leftmargin=*]
    \item Degree query: Given $u \in V$, the oracle returns the degree of $u$.
    \item Edge query: Given $u \in V$ and index $i$, the oracle returns the $i$-th neighbor of $u$, or $\bot$ if the edge does not exist.
    \item Adjacency query: Given $u, v \in V$, the oracle returns whether $(u ,v) \in E$.
\end{enumerate}

In the $\textsc{Min-Cut}$ problem, our goal is to estimate the global minimum cut up to a $(1 \pm \eps)$-factor using these local queries. The complexity of the problem is measured by the number of queries, and we want to use as few queries as possible. For this problem we focus on undirected graphs. 

Previous work \cite{eden2017lower} showed an $\Omega(\frac{m}{k})$ query complexity lower bound, where $k$ is the size of the minimum cut.
The main open question is what the dependence on $\eps$ should be. There is also an $\tilde{O}(\frac{m}{k \poly(\eps)})$ upper bound in \cite{globalmincut}, and a natural question is to close this gap.

\subsection{Our Results}
We resolve the main open questions mentioned above.

\medskip
{\noindent \bf Cut Sketch for Balanced (Directed) Graphs.}
We study the space complexity of $(1\pm\eps)$ cut sketches for $n$-node $\beta$-balanced (directed) graphs.
Previous work~\cite{IT18,CCP+21} gave an $\tilde{O}(n\beta/\eps^2)$ upper bound in the for-all model and an $\tilde{O}(n\sqrt{\beta}/\eps)$ upper bound in the for-each model, along with an $\Omega(n\beta/\eps)$ lower bound and an $\Omega(n\sqrt{\beta/\eps})$ lower bound, respectively.

We close these gaps and resolve the dependence on $\eps$, improving the lower bounds to match the upper bounds for all parameters $n, \beta$, and $\eps$ (up to logarithmic factors). Formally, we have:

\begin{restatable}[For-Each Cut Sketch for Balanced Graphs]{theorem}{ForEachCutSketch}
\label{thm:for_each}
Let $\beta \ge 1$ and $0 < \eps < 1$. Assume $\sqrt{\beta}/\eps \le n / 2$. Any $(1 \pm \eps)$ for-each cut sketching algorithm for $\beta$-balanced $n$-node graphs must output $\tilde{\Omega}(n \sqrt{\beta}/ \eps)$ bits.
\end{restatable}

\begin{restatable}[For-All Cut Sketch for Balanced Graphs]{theorem}{ForAllCutSketch}
\label{thm:for_all}
Let $\beta \ge 1$ and $0 < \eps < 1$. Assume $\beta/\eps^2 \le n / 2$.
Any $(1 \pm \eps)$ for-all cut sketching algorithm for $\beta$-balanced $n$-node graphs must output $\Omega(n \beta/ \eps^2)$ bits.
\end{restatable}

\medskip
{\noindent \bf Query Complexity of Min-Cut in the Local Query Model.}
We study the problem of $(1 \pm \eps)$-approximating the (undirected) global minimum cut in a local query model, where we can only access the graph via degree, edge, and adjacency queries.

We close the gap on the $\eps$ dependence in the query complexity of this problem by proving a tight $\Omega(\min\{m, \frac{m}{\eps^2 k}\})$ lower bound, where $m$ is the number of edges and $k$ is the size of the minimum cut.
This improves the previous $\Omega\big(\frac{m}{k}\big)$ lower bound in \cite{eden2017lower}. Formally, we have:

\begin{restatable}[Approximating Min-Cut using Local Queries]{theorem}{MinCutLocalQuery}
\label{thm:min_cut}
Any algorithm that estimates the size of the global minimum cut of a graph $G$ up to a $(1 \pm \eps)$ factor requires $\Omega(\min\{m, \frac{m}{\eps^2k}\})$ queries in expectation in the local query model, where $m$ is the number of edges in $G$ and $k$ is the size of the minimum cut.
\end{restatable}

We also show that with a slight modification, the $\tilde{O}(\frac{m}{k \poly(\eps)})$ query complexity upper bound in~\cite{globalmincut} can be improved to $\tilde{O}\left(\frac{m}{\eps^2k}\right)$, which implies that our lower bound is tight (up to logarithmic factors).

\subsection{Our Techniques}
A common technique we use for the different problems is communication complexity games that involve the approximation parameter $\eps$. 
For example, suppose Alice has a bit string $s$ of length $(1/\eps^2)$, and she can encode $s$ into a graph $G$ such that, if she sends Bob a $(1\pm\eps)$ (for-each or for-all) cut sketch to Bob, then Bob can recover a specific bit of $s$ with high constant probability.
By communication complexity lower bounds, we know Alice must send $\Omega(1/\eps^2)$ bits to Bob, which gives a lower bound on the size of the cut sketch.

\medskip
{\noindent \bf For-Each Cut Sketch Lower Bound.} Let $k = \sqrt{\beta}/\eps$. At a high level, we partition the $n$ nodes into $n/(2k)$ sub-graphs, where each sub-graph is a $k$-by-$k$ bipartite graph with two parts $L$ and $R$. We then divide $L$ and $R$ into $\sqrt{\beta}$ disjoint clusters $|L_1| = |L_2| = \ldots = |L_{\sqrt{\beta}}| = 1/\eps$ and $|R_1| = |R_2| = \ldots = |R_{\sqrt{\beta}}| = 1/\eps$. For every cluster pair $L_i$ and $R_j$, there are a total of $1/\eps^2$ edges. Intuitively, we wish to encode a bit string $s \in \{-1, 1\}^{1/\eps^2}$ into forward edges (left to right) each with weight $\Theta(1)$, and add backward edges (right to left) each with weight ${1}/{\beta}$ so that the graph $\beta$-balanced.
If we could approximately decode this string from a for-each cut sketch, then we would get an $\Omega((n/k) \cdot (\sqrt{\beta})^2\cdot (1/\eps)^2) = \Omega(n\sqrt{\beta}/\eps)$ lower bound.

However, if we use a simple encoding method~\cite{ACK+15,CCP+21} where each bit $s_i$ is encoded into one edge $(u, v)$ (e.g., with weight $1$ or $2$) and query the edges leaving $S = \{u\} \cup (R \setminus \{v\})$, then the $(k-1)^2 = \Omega(\beta/\eps^2)$ backward edges with weight $1/\beta$ will cause the cut value to be $\Omega(1/\eps^2)$.
The $(1\pm\eps)$ cut sketch will have additive error $\Omega(1/\eps) \gg \Theta(1)$, which will obscure $s_i = \{-1, 1\}$. To address this, we instead encode $1/\eps^2$ bits of information across $1/\eps^2$ edges simultaneously. When we want to decode a specific bit $s_i$, we query the (directed) cut values between two {\em carefully designed} subsets $A \in L_i$ and $B \in R_j$. The key idea of our construction is that, although each edge in $A \times B$ is used to encode many bits of $s$, the encoding of different bits of $s$ is {\em never too correlated}: while encoding other bits does affect the total weight from $A$ to $B$, this effect is similar to adding noise which only varies the total weight from $A$ to $B$ by a small amount. 


\medskip
{\noindent \bf For-All Cut Sketch Lower Bound.} Let $k = \beta/\eps^2$. At a high level, we partition the $n$ nodes into $n/(2k)$ sub-graphs, where each sub-graph is a $k$-by-$k$ bipartite graph with two parts $L$ and $R$. Let $L = \{\ell_1, \ldots, \ell_k\}$. We partition $R$ into $\beta$ disjoint clusters $|R_1| = \ldots = |R_\beta| = 1/\eps^2$.
We use edges from $\ell_i$ to $R_j$ to encode a bit string $s \in \{0, 1\}^{1/\eps^2}$ by setting the weight of each forward edge to $1$ or $2$, and adding a backward edge of weight ${1}/{\beta}$ to balance the graph.

We can show that the following problem requires $\Omega(1/\eps^2)$ bits of communication:
Consider $\ell_i \in L$ and a random subset $T \subset R_j$ where $|T| = \frac{|R_j|}{2}$.
Let $N(\ell_i)$ denote $\ell_i$'s neighbors $v$ such that $(\ell_i, v)$ has weight $2$, which is uniformly random if $s$ is uniformly random.
The problem is to decide whether $|N(\ell_i) \cap T| \ge \frac{1}{4\eps^2} + \frac{c}{2\eps}$ or $|N(\ell_i) \cap T| \le \frac{1}{4\eps^2} - \frac{c}{2\eps}$ for a sufficiently small constant $c > 0$.
Intuitively, the graph encodes a $(k \beta)$-fold version of this communication problem, which implies an $\Omega((n/k) \cdot k \beta \cdot (1/\eps)^2) = \Omega(n\beta/\eps^2)$ lower bound.

We need to show that Bob can distinguish between the two cases of $|N(\ell_i) \cap T|$ given a for-all cut sketch.
However, there are some challenges. The difference between the two cases is $\Theta(1/\eps)$ while the natural cut to query $S = \{\ell_i\} \cup (R \setminus T)$ has value $\Omega(\beta/\eps^4)$.
The $(1\pm\eps)$ cut sketch will have additive error $\Omega(\beta/\eps^3) \gg \Theta(1/\eps)$, which is too much.
To overcome this, note that we have not used the property that the for-all cut sketch preserves all cuts.
We make use of the following crucial observation in~\cite{ACK+15}:
In expectation, roughly half of the nodes $\ell_i \in L$ satisfy $|N(\ell_i) \cap T| \ge \frac{1}{4\eps^2} + \frac{c}{2\eps}$ because $c$ is small. If Bob enumerates all subsets $Q \subset L$ of size $\frac{|L|}{2}$, he will eventually get lucky and find a set $Q$ that contains almost all such nodes.
Since there are roughly $\frac{|L|}{2} = \frac{\beta}{2\eps^2}$ such nodes, the $(c/\eps)$ bias per node will contribute $\Omega(c \beta / \eps^3)$ in total, which is enough to be detected even under an $O(\beta/\eps^3)$ additive error.

\medskip
{\noindent \bf Query Complexity of Min-Cut in the Local Query Model.}
We prove our lower bound using communication complexity, but unlike previous work~\cite{eden2017lower}, we consider the following $\textsc{2SUM}$ problem~\cite{2sum}:
Given $2t$ length-$L$ binary strings $(x^1, x^2, \ldots, x^t)$ and $(y^1, y^2, \ldots, y^t)$, we want to approximate the value of $\sum_{i \in [t]} \textsc{DISJ}(x^i, y^i)$ up to a $\sqrt{t}$ additive error, with the promise that at least a constant fraction of the $(x^i, y^i)$ satisfy $\textsc{INT}(x^i, y^i) = \alpha$ while the remaining pairs satisfy $\textsc{INT}(x^i, y^i) = 0$ or $\alpha$.
Here $\textsc{INT}(x, y) = \sum_{i=1}^L x_i \wedge y_i$ is the number of indices where $x$ and $y$ are both $1$, and $\textsc{DISJ}(x, y)$ is the set-disjointness problem, i.e., $\textsc{DISJ}(x, y) = 1$ if $\textsc{INT}(x, y) = 0$ and $\textsc{DISJ}(x, y) = 0$ otherwise. The parameters $L, t$, and $\alpha$ will be chosen later. 

We construct our graph $G_{x,y}$ based on the vectors $x^i$ and $y^i$ in a way inspired by~\cite{eden2017lower}. We then give a careful analysis of the size of the minimum cut of $G_{x,y}$, and show that under certain conditions, the size of the minimum cut is exactly $2\sum_{i \in [t]}\textsc{INT}(x^i, y^i)$.
Consequently, a $(1 \pm \eps)$-approximation of the minimum cut yields an approximation of $\sum_{i \in [t]} \textsc{DISJ}(x^i, y^i)$ up to a $\sqrt{\eps}$ additive error, which implies the desired lower bound.

\section{Preliminaries}

Let $G = (V, E, w)$ be a weighted (directed) graph with $n$ vertices and $m$ edges, where each edge $e \in E$ has weight $w_e\geq 0$. We write $G = (V, E)$ if $G$ is unweighted and leave out $w$.
For two sets of nodes $S, T \subseteq V$, let $E(S, T) = \{(u, v) \in E: u \in S, v \in T\}$ denote the set of edges from $S$ to $T$.
Let $w(S, T) = \sum_{e \in E(S,T)} w_e$ denote the total weight of edges from $S$ to $T$.
For a node $u \in V$ and a set of nodes $S \subseteq V$, we write $w(u, S)$ for $w(\{u\}, S)$.

We write $[n]$ for $\{1, \ldots, n\}$. We use $\mathbf{1}$ to denote the all-ones vector.
For a vector $v$, we write $\norm{v}_2$ and $\norm{v}_\infty$ for the $\ell_2$ and $\ell_\infty$ norm of $x$ respectively.
For two vectors $u, v \in \R^{n}$, let $u \otimes v \in \R^{n^2}$ be the tensor product of $u$ and $v$. Given a matrix $A$, we use $A_i$ to denote the $i$-th row of $A$.

\paragraph{Directed Cut Sketches.}
We start with the definitions of $\beta$-balanced graphs, for-all and for-each cut sketches~\cite{BK96, ST11, ACK+15, CCP+21}.

We say a directed graph is balanced if all cuts have similar values in both directions.

\begin{definition}[$\beta$-Balanced Graphs]
A strongly connected directed graph $G = (V, E, w)$ is \emph{$\beta$-balanced} if, for all $\varnothing \subset S \subset V$, 
it holds that $w(S, V \setminus S) \le \beta \cdot w(V \setminus S, S)$.
\end{definition}

We say $\sk(G)$ is a for-all cut sketch if the value of all cuts can be approximately recovered from it.
Note that $\sk(G)$ is not necessarily a graph and can be an arbitrary data structure.

\begin{definition}[For-All Cut Sketch]
\label{def:for-all}
Let $0 < \eps < 1$.
We say $\mathcal{A}$ is a $(1 \pm \eps)$ for-all cut sketching algorithm if there exists a recovering algorithm $f$ such that, given a directed graph $G = (V, E, w)$ as input, $\mathcal{A}$ can output a sketch $\sk(G)$ such that, with probability at least $2/3$, for all $\varnothing \subset S \subset V$:
\[
(1-\eps) \cdot w(S, V \setminus S) \le f(S, \sk(G)) \le (1+\eps) \cdot w(S, V \setminus S).
\]
\end{definition}

Another notion of cut approximation is that of a ``for-each'' cut sketch, which requires that the value of each individual cut is preserved with high constant probability, rather than approximating the values of all cuts simultaneously.

\begin{definition}[For-Each Cut Sketch]
\label{def:for-each}
Let $0 < \eps < 1$.
We say $\mathcal{A}$ is a $(1\pm\eps)$ for-each cut sketching algorithm if there exists a recovering algorithm $f$ such that, given a directed graph $G = (V, E, w)$ as input, $\mathcal{A}$ can output a sketch $\sk(G)$ such that, for each $\varnothing \subset S \subset V$, with probability at least $2/3$,
\[
(1-\eps) \cdot w(S, V \setminus S) \le f(S, \sk(G)) \le (1+\eps) \cdot w(S, V \setminus S).
\]
\end{definition}

In Definitions~\ref{def:for-all}~and~\ref{def:for-each}, the sketching algorithm $\mathcal{A}$ and the recovering algorithm $f$ can be randomized, and the probability is over the randomness in $\mathcal{A}$ and $f$.
\section{For-Each Cut Sketch}

In this section, we prove an $\Omega(n\sqrt{\beta}/\eps)$ lower bound on the output size of $(1 \pm \eps)$ for-each cut sketching algorithms (Definition~\ref{def:for-each}).

\ForEachCutSketch*

Our result uses the following communication complexity lower bound for a variant of the Index problem, where Alice and Bob's inputs are random.

\begin{lemma}[\cite{KNR01}]
\label{lem:index}
Suppose Alice has a uniformly random string $s \in \{-1, 1\}^n$ and Bob has a uniformly random index $i \in [n]$. If Alice sends a single (possibly randomized) message to Bob, and Bob can recover $s_{i}$ with probability at least $2/3$ (over the randomness in the input and their protocol), then Alice must send $\Omega(n)$ bits to Bob.
\end{lemma}

Our lower-bound construction relies on the following technical lemma.

\begin{lemma}
\label{lem:matrix}
    For any integer $k \ge 1$, there exists a matrix $M \in \{-1, 1\}^{(2^k - 1)^2 \times 2^{2k}}$ such that:
    \begin{enumerate}
        \item $\inner{M_t, \mathbf{1}} = 0$ for all $t \in [(2^k - 1)^2]$. 
        \item $\inner{M_t, M_{t'}} = 0$ for all $1 \le t < t' \le (2^k - 1)^2$.
        \item For all $t \in [(2^k - 1)^2]$, the t-th row of $M$ can be written as $M_t = u \otimes v$ where $u, v \in \{-1, 1\}^{2^k}$ and $\inner{u, \mathbf{1}} = \inner{v, \mathbf{1}} = 0$.
    \end{enumerate}
\end{lemma}
\begin{proof}
Our construction is based on the Hadamard matrix $H = H_{2^k} \in \{-1, 1\}^{2^k \times 2^k}$. Recall that the first row of $H$ is the all-ones vector and that $\inner{H_i, H_j} = 0$ for all $i \ne j$.
For every $2 \le i, j \le 2^k$, we add $H_i \otimes H_j \in \{-1, 1\}^{2^{2k}}$ as a row of $M$, so $M$ has $(2^k - 1)^2$ rows.

Condition~(3) holds because $\inner{H_i, \mathbf{1}} = \inner{H_j, \mathbf{1}} = 0$ for all $i, j \ge 2$.
For Conditions~(1) and~(2), note that for any vectors $u, v, w$, and $z$, we have $\inner{u\otimes v, w \otimes z} = \inner{u, w} \inner{v, z}$.
Using this fact, Condition~(1) holds because $\inner{M_t, \mathbf{1}} = \inner{H_i \otimes H_j, \mathbf{1} \otimes \mathbf{1}} = \inner{H_i, \mathbf{1}} \inner{H_j, \mathbf{1}} = 0$,
and Condition~(2) holds because $(i,j) \neq (i',j')$ and thus $\inner{M_t, M_{t'}} = \inner{H_i \otimes H_j, H_{i'} \otimes H_{j'}} = \inner{H_{i}, H_{i'}}\inner{H_{j}, H_{j'}} = 0$.
\end{proof}

We first prove a lower bound for the special case $n = \Theta(\sqrt{\beta}/\eps)$.
Our proof for this special case introduces important building blocks for proving the general case $n = \Omega(\sqrt{\beta}/\eps)$.

\begin{lemma}
\label{lem:for_each_unit}
    Suppose $n = \Theta(\sqrt{\beta}/\eps)$.
    Any $(1 \pm \eps)$ for-each cut sketching algorithm for $\beta$-balanced $n$-node graphs must output $\tilde{\Omega}(n \sqrt{\beta}/ \eps) = \tilde{\Omega}(\beta/ \eps^2)$ bits.
\end{lemma}
    At a high level, we reduce the Index problem (Lemma~\ref{lem:index}) to the for-each cut sketching problem.
    Given Alice's string $s$, we construct a graph $G$ to encode $s$, such that Bob can recover any single bit in $s$ by querying $O(1)$ cut values of $G$.
    Our lower bound (Lemma~\ref{lem:for_each_unit}) then follows from the communication complexity lower bound of the Index problem (Lemma~\ref{lem:index}), because Alice can run a for-each cut sketching algorithm and send the cut sketch to Bob, and Bob can successfully recover the $O(1)$ cut values with high constant probability.

\begin{proof}[Proof of Lemma~\ref{lem:for_each_unit}]
    We reduce from the Index problem.
    Let $s \in \{-1, 1\}^{\beta(\frac{1}{\eps}-1)^2}$ denote Alice's random string.

    \medskip
    {\noindent \bf Construction of $G$.}
    We construct a directed complete bipartite graph $G$ to encode $s$.
    Let $L$ and $R$ denote the left and right nodes of $G$, where $|L| = |R| = \sqrt{\beta}/\eps$.
    We partition $L$ into $\sqrt{\beta}$ disjoint blocks $L_1, \ldots, L_{\sqrt{\beta}}$ of equal size, and similarly partition $R$ into $R_1, \ldots, R_{\sqrt{\beta}}$.
    We divide $s$ into $\beta$ disjoint strings $s_{i,j} \in \{-1, 1\}^{(\frac{1}{\eps}-1)^2}$ of the same length.
    We will encode $s_{i,j}$ using the edges from $L_i$ to $R_j$.
    Note that the encoding of  each $s_{i,j}$ is independent since $E(L_i, R_j) \cap E(L_{i'}, R_{j'}) = \varnothing$ for $(i, j) \neq (i', j')$.
    
    We fix $i$ and $j$ and focus on the encoding of $s_{i,j}$.
    Note that $|L_i| = |R_j| = 1/\eps$.
    We refer to the edges from $L_i$ to $R_j$ as {\em forward} edges and the edges from $R_j$ to $L_i$ as {\em backward} edges.
    Let $w \in \mathbb{R}^{1/\eps^2}$ denote the weights of the forward edges, which we will choose soon.
    Every backward edge has weight $1/\beta$.
    
    Let $z = s_{i,j} \in \{-1, 1\}^{(\frac{1}{\eps}-1)^2}$.
    Assume w.l.o.g. that $1/\eps = 2^k$ for some integer $k$.
    Consider the vector $x = \sum_{t=1}^{(\frac{1}{\eps}-1)^2} z_t M_t \in \R^{1/\eps^2}$ where $M$ is the matrix in Lemma~\ref{lem:matrix} with $2^k = 1/\eps$.
    Because $z_t \in \{-1, 1\}$ is uniformly random, each coordinate of $x$ is a sum of $O(1/\eps^2)$ i.i.d. random variables of value $\pm 1$.
    By the Chernoff bound and the union bound, we know that with probability at least $99/100$, $\norm{x}_{\infty} \le c_1 \ln(1/\eps)/\eps$ for some constant $c_1 > 0$. If this happens, we set $w = \eps x + 2 c_1 \ln(1/\eps) \mathbf{1}$, so that each entry of $w$ is between $c_1 \ln(1/\eps)$ and $3 c_1 \ln(1/\eps)$. Otherwise, we set $w = 2 c_1 \ln(1/\eps) \mathbf{1}$ to indicate that the encoding failed.

    We first verify that $G$ is $O(\beta \log(1/\eps))$-balanced.
    This is because every edge has a reverse edge with similar weight:
    For every $u \in L$ and $v \in R$, the edge $(u, v)$ has weight $\Theta(\log(1/\eps))$, while the edge $(v, u)$ has weight $1/\beta$.

    We will show that given a $(1 \pm \frac{c_2 \eps}{\ln(1/\eps)})$ cut sketch for some constant $c_2 > 0$, Bob can recover a specific bit of $z$ using $4$ cut queries.
    By Lemma~\ref{lem:index}, this implies an $\Omega(\beta/\eps^2) = \tilde{\Omega}(\beta'/\eps'^2)$ lower bound for cut sketching algorithms for $\beta' = O(\beta \log(1/\eps))$ and $\eps' = c_2\eps / \ln(1/\eps)$.

    \medskip
    {\noindent \bf Recovering a bit in $s$ from a for-each cut sketch of $G$.}
    Suppose Bob wants to recover a specific bit of $s$, which belongs to the substring $z = s_{i,j}$ and has an index $t$ in $z$.
    We assume that $z$ is successfully encoded by the subgraph between $L_i$ and $R_j$.
    
    For simplicity, we index the nodes in $L_i$ as $1, \ldots, (1/\eps)$ and similarly for $R_j$.
    We index the forward edges $(u, v)$ in alphabetical order, first by $u \in L_i$ and then by $v \in R_j$.
    Under this notation, $\inner{w, \mathbf{1}_A \otimes \mathbf{1}_B}$ gives the total weight $w(A, B)$ of forward edges from $A$ to $B$, where $\mathbf{1}_A, \mathbf{1}_B \in \{0,1\}^{1/\eps}$ are the indicator vectors of $A \subset L_i$ and $B \subset R_j$.
    
    The crucial observation is that, given a cut sketch of $G$, Bob can approximate $\inner{w, M_t}$ using $4$ cut queries. By Lemma~\ref{lem:matrix}, $M_t = h_A \otimes h_B$ for some $h_A, h_B \in \{-1, 1\}^{1/\eps}$.
    Let $A \subset L_i$ be the set of nodes $u \in L_i$ with $h_A(u) = 1$.
    Let $B \subset R_i$ be the set of nodes $v \in R_j$ with $h_B(v) = 1$.
    Let $\bar{A} = L_i \setminus A$ and $\bar{B} = R_j \setminus B$.
     \begin{align*}
     \inner{w, M_t}
     &= \inner{w, h_A \otimes h_B}
     = \inner{w, (\mathbf{1}_A - \mathbf{1}_{\bar A}) \otimes (\mathbf{1}_B - \mathbf{1}_{\bar B})} \\
     &= w(A, B) - w(\Bar{A}, B) - w(A, \Bar{B}) + w(\Bar{A}, \Bar{B}) \;.
     \end{align*}

    To approximate the value of $w(A, B)$ (and similarly $w(\Bar{A}, B)$, $w(A, \Bar{B})$, $w(\Bar{A}, \Bar{B})$), Bob can query $w(S, V \setminus S)$ for $S = A \cup(R \setminus B)$.
    Consider the edges from $S$ to $(V \setminus S)$: the forward edges are from $A$ to $B$, each with weight $\Theta(\log(1/\eps))$; and the backward edges are from $(R \setminus B)$ to $(L \setminus A)$, each with weight $1/\beta$.
    See Figure~\ref{fig:for_Each} as an example.
    
    \begin{figure}[ht]
        \centering
        \includegraphics[width=0.35\linewidth]{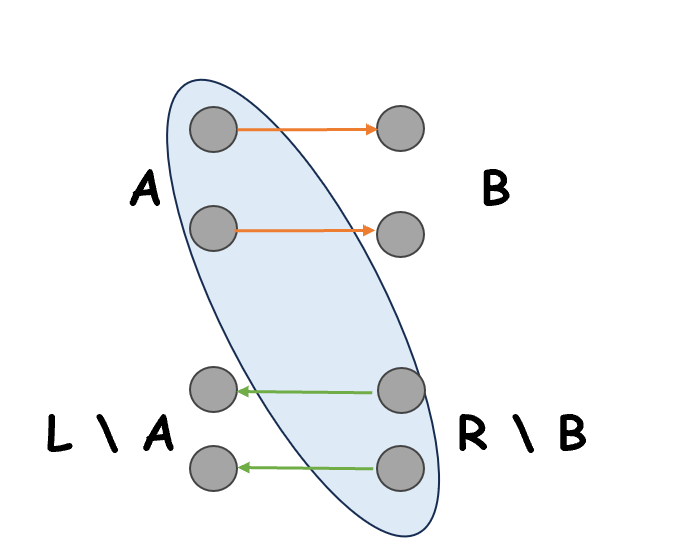}
        \caption{For $S = A \cup(R \setminus B)$, the (directed) edges from $S$ to $(V \setminus S)$ consist of the following: the forward edges from $A$ to $B$, each with weight $\Theta(\log(1/\eps))$, and the backward edges from $(R \setminus B)$ to $(L \setminus A)$, each with weight $1/\beta$.}
        \label{fig:for_Each}
    \end{figure}
    
    By Lemma~\ref{lem:matrix}, $\inner{h_A, \mathbf{1}} = \inner{h_B, \mathbf{1}} = 0$, so $|A| = |B| = \frac{|L_i|}{2} = \frac{|R_j|}{2} = \frac{1}{2\eps}$.
    The total weight of the forward edges is $\Theta(\log(1/\eps) / \eps^2)$, and the total weight of the backward edges is $\bigl(\frac{\sqrt{\beta}}{\eps}-\frac{1}{2\eps}\bigr)^2\frac{1}{\beta} = \Theta(1/\eps^2)$, so the cut value $w(S, V \setminus S)$ is $\Theta(\log(1/\eps)/\eps^2)$.
    Given a $(1 \pm \frac{c_2 \eps}{\ln(1/\eps)})$ for-each cut sketch, Bob can obtain a $(1 \pm \frac{c_2\eps}{\log(1/\eps)})$ multiplicative approximation of $w(S, V \setminus S)$, which has $O(c_2/\eps)$ additive error.
    After subtracting the total weight of backward edges, which is fixed, Bob has an estimate of $w(A, B)$ with $O(c_2/\eps)$ additive error.
    Consequently, Bob can approximate $\inner{w, M_t}$ with $O(c_2/\eps)$ additive error using $4$ cut queries.

    Now consider $\inner{w, M_t}$. By Lemma~\ref{lem:matrix}, $\inner{M_t, \mathbf{1}} = 0$ and the rows of $M$ are orthogonal,
     \begin{align*}
     \inner{w, M_t} = \inner{\eps x, M_t} &= \eps \inner{\sum_{t'} z_{t'} M_{t'}, M_t}
     = \eps z_t \norm{M_t}_2^2 = \frac{z_t}{\eps} \;.    
     \end{align*}
    We can see that, for a sufficiently small universal constant $c_2$, Bob can distinguish whether $z_t = 1$ or $z_t = -1$ based on an $O(c_2/\eps)$ additive approximation of $\inner{w, M_t}$.
    
    Bob's success probability is at least $0.95$, because the encoding of $z$ fails with probability at most $0.01$, and each of the $4$ cut queries fails with probability at most $0.01$.~\footnote{The success probability of a cut query given a for-each cut sketch (Definition~\ref{def:for-each}) can be boosted from $2/3$ to $99/100$, e.g., by running the sketching and recovering algorithms $O(1)$ times and taking the median. This increases the length of Alice's message by a constant factor, which does not affect our asymptotic lower bound.}
\end{proof}

We next consider the case with general values of $n, \beta$, and $\eps$, and prove Theorem~\ref{thm:for_each}.

\begin{proof}[Proof of Theorem~\ref{thm:for_each}]
    Let $k = \sqrt{\beta} / \eps$. We assume w.l.o.g. that $k$ is an integer, $n$ is a multiple of $k$, and $(1/\eps)$ is a power of $2$. Suppose Alice has a random string $s \in \{-1, 1\}^{\Omega(nk)}$. We will show that $s$ can be encoded into a graph $G$ such that
    
    {\em (i)} $G$ has $n$ nodes and is $O(\beta \log(1/\eps))$-balanced, and
    
    {\em (ii)} Given a $(1 \pm \frac{c_2\eps} {\ln(1/\eps)})$ for-each cut sketch of $G$ and an index $q$, where $c_2 > 0$ is a sufficiently small universal constant, Bob can recover $s_q$ with probability at least $2/3$.
    
    Consequently, by Lemma~\ref{lem:index}, any for-each cut sketching algorithm must output $\Omega(nk) = \Omega(n \sqrt{\beta}/\eps) \allowbreak = \tilde{\Omega}(n \sqrt{\beta'}/\eps')$ bits for $\beta' = O(\beta \log(1/\eps))$ and $\eps' = c_2 \eps / \ln(1/\eps)$.

    We first describe the construction of $G$.
    We partition the $n$ nodes into $\ell = n / k \ge 2$ disjoint sets $V_1, \ldots, V_\ell$, each containing $k$ nodes.
    Let $s$ be Alice's random string with length $\beta (\frac{1}{\eps} - 1)^2 (\ell-1) = \Omega(k^2 \ell)= \Omega(nk)$.
    We partition $s$ into $(\ell - 1)$ strings $(s_i)_{i=1}^{\ell - 1}$, with $k^2$ bits in each substring.
    We then follow the same procedure as in Lemma~\ref{lem:for_each_unit} to encode $s_i$ into a complete bipartite graph between $V_i$ and $V_{i+1}$. Notice that we have $|s_i| = \beta (\frac{1}{\eps} - 1)^2 $ and $|V_i| = |V_{i+1}| = \sqrt{\beta}/\eps$, which is the same setting as in Lemma~\ref{lem:for_each_unit}. 

    We can verify that $G$ is $O(\beta\log(1/\eps))$-balanced.
    This is because every edge $e$ has a reverse edge whose weight is at most $O(\beta\log(1/\eps))$ times the weight of $e$.
    For every $u \in V_i$ and $v \in V_{i+1}$, the edge $(u, v)$ has weight $\Theta(\log(1/\eps))$, while the edge $(v, u)$ has weight $1/\beta$.

    We next show that Bob can recover the $q$-th bit of $s$. Suppose Bob's index $q$ belongs to the substring $s_i$ which is encoded by the subgraph between $V_{i}$ and $V_{i+1}$.
    Similar to the proof of Lemma~\ref{lem:for_each_unit}, Bob only needs to approximate $w(A, B)$ for $4$ pairs of $(A, B)$ with $O(1/\eps)$ additive error, where $A \subset V_i$, $B \subset V_{i+1}$, and $|A| = |B| = \frac{1}{2\eps}$. To achieve this, Bob can query the cut value $w(S, V \setminus S)$ for $S = A \cup \left(V_{i + 1} \setminus B\right) \bigcup_{j = i + 2}^{\ell} V_j$. 
    The edges from $S$ to $(V \setminus S)$ are:
\begin{itemize}[leftmargin=*]
    \item $\frac{1}{4\eps^2}$ forward edges from $A$ to $B$, each with weight $\Theta(\log(1/\eps))$.
    \item $\bigl(\frac{\sqrt{\beta}}{\eps}-\frac{1}{2\eps}\bigr)^2$ backward edges from $(V_{i+1} \setminus B)$ to $(V_{i} \setminus A)$, each with weight $\frac{1}{\beta}$.
    \item $\frac{\sqrt{\beta}}{2\eps^2}$ backward edges from $A$ to $V_{i - 1}$ when $i \ge 2$, each with weight $\frac{1}{\beta}$.
\end{itemize}

The cut value $w(S, V \setminus S)$ is $\Theta(\log(1/\eps)/\eps^2)$.
Consequently, given a $(1 \pm \frac{c_2\eps} {\ln(1/\eps)})$ for-each cut sketch, after subtracting the fixed weight of the backward edges, Bob can approximate $w(A, B)$ with $O(c_2/\eps)$ additive error.
Similar to the proof of Lemma~\ref{lem:for_each_unit}, for sufficiently small constant $c_2 > 0$, repeating this process for $4$ different pairs of $(A, B)$ will allow Bob to recover $s_q \in \{-1, 1\}$.
\end{proof}

\section{For-All Cut Sketch}

In this section, we prove an $\Omega(n\beta/\eps^2)$ lower bound on the output size of $(1 \pm \eps)$ for-all cut
sketching algorithms (Definition~\ref{def:for-all}).

\ForAllCutSketch*

Our proof is inspired by \cite{ACK+15} and uses the following communication complexity lower bound for an $n$-fold version of the Gap-Hamming problem.

\begin{lemma}[\cite{ACK+15}]
\label{lem:gap_hamming}
Consider the following distributional communication problem:
Alice has $h$ strings $s_1,\ldots,s_{h} \in \{0, 1\}^{1/\eps^2}$
of Hamming weight $\frac{1}{2\eps^2}$.
Bob has an index $i \in [h]$ and a string $t\in \{0, 1\}^{1/\eps^2}$
of Hamming weight $\frac{1}{2\eps^2}$, 
drawn as follows: 
\begin{enumerate}
\item $i$ is chosen uniformly at random;
\item every $s_{i'}$ for $i'\neq i$ is chosen uniformly at random;
\item $s_i$ and $t$ are chosen uniformly at random, conditioned on 
their Hamming distance $\Delta(s_i,t)$ being, with equal probability,  
either $\geq \frac{1}{2\eps^2} + \frac{c}{\eps}$ 
or     $\leq \frac{1}{2\eps^2} - \frac{c}{\eps}$ for some universal constant $c > 0$.
\end{enumerate}
Consider a (possibly randomized) one-way protocol, in which Alice sends Bob a message, and Bob then determines with success probability at least $2/3$ whether $\Delta(s_i,t)$
is $\geq \frac{1}{2\eps^2} + \frac{c}{\eps}$ or $\leq \frac{1}{2\eps^2} - \frac{c}{\eps}$.
Then Alice must send $\Omega(h/\eps^2)$ bits to Bob. 
\end{lemma}

Before proving Theorem~\ref{thm:for_all}, we first consider the special case $n = \Theta(\beta/\eps^2)$.

\begin{lemma}
\label{lem:for_all_unit}
    Suppose $n = \Theta(\beta/\eps^2)$. Any $(1 \pm \eps)$ for-all cut sketching algorithm for $\beta$-balanced $n$-node graphs must output $\Omega(n \beta/ \eps^2) = \Omega(\beta^2/ \eps^4)$ bits.
\end{lemma}
We reduce the distributional Gap-Hamming problem (Lemma~\ref{lem:gap_hamming}) to the for-all cut sketching problem.
    Suppose Alice has $h$ strings $s_1, s_2, \ldots, s_h \in \{0, 1\}^{1/\eps^2}$ where $h = \beta^2 / \eps^2$, and Bob has an index $i \in [h]$ and a string $t \in \{0, 1\}^{1/\eps^2}$.
    We construct a graph $G$ to encode $s_1, s_2, \ldots, s_h$, such that given a for-all cut sketch of $G$, Bob can determine whether $\Delta(s_i, t) \ge \frac{1}{2\eps^2} + \frac{c}{\eps}$ or $\Delta(s_i, t) \le \frac{1}{2\eps^2} - \frac{c}{\eps}$ with high constant probability.
    Our lower bound then follows from Lemma~\ref{lem:gap_hamming}.

\medskip
{\bf \noindent Construction of $G$.} We construct a directed complete bipartite graph $G$.
Let $L$ and $R$ denote the left and right nodes of $G$, where $|L| = |R| = \beta / \eps^2$. We partition $R$ into $\beta$ disjoint sets with $|R_1| = \ldots = |R_\beta| = 1/\eps^2$.

Consider the distributional Gap-Hamming problem in Lemma~\ref{lem:gap_hamming} with $h = \beta^2 / \eps^2$. We re-index Alice's $(\beta^2 / \eps^2)$ strings as $s_{i, j}$, where $i \in [\beta / \eps^2]$ and $j \in [\beta]$. Let $\ell_1, \ell_2, \ldots, \ell_{\beta/\eps^2}$ be the nodes in $L$. We encode $s_{i, j} \in \{0, 1\}^{1/\eps^2}$ using the edges from $\ell_i$ to $R_j$: For node $\ell_i$ and the $v$-th node in $R_j$, the {\em forward} edge $(\ell_i, v)$ has weight $s_{i, j}(v) + 1$, and the {\em backward} edge $(v, \ell_i)$ has weight $1/\beta$. Note that the encoding of each $s_{i,j}$ is independent since $E(\ell_i, R_j) \cap E(\ell_{i'}, R_{j'}) = \varnothing$ for $(i, j) \neq (i', j')$.

\medskip
{\bf \noindent Determining $\Delta(s_{i, j}, t)$ from a for-all cut sketch of $G$.} 
Suppose Bob's input (after re-indexing) is $1 \le i \le \beta / \eps^2$, $1 \le j \le \beta$, and $t \in \{0, 1\}^{1/\eps^2}$. Bob wants to decide whether $\Delta(s_{i, j}, t) \ge \frac{1}{2\eps^2} + \frac{c}{\eps}$ or $\Delta(s_{i, j}, t) \le \frac{1}{2\eps^2} - \frac{c}{\eps}$.

Let $N(\ell_i)$ denote the set of nodes $v \in R_j$ where the forward edge $(\ell_i, v)$ has weight $2$, which corresponds to the positions of $1$ in $s_{i,j}$. Let $T$ be the set of nodes $v \in R_j$ such that $t(v) = 1$.
\[
\Delta(s_{i, j}, t) = |N(\ell_i) \setminus T| + |T \setminus N(\ell_i)| = |N(\ell_i)| + |T| - 2 |N(\ell_i) \cap T| = \frac{1}{\eps^2} - 2 | N(\ell_i) \cap T| \; .
\]
Hence, to determine whether $\Delta(s_{i, j}, t) \le \frac{1}{2\eps^2} - \frac{c}{\eps}$ or $\Delta(s_{i, j}, t) \ge \frac{1}{2\eps^2} + \frac{c}{\eps}$, Bob only needs to decide whether $|N(\ell_i) \cap T| \ge \frac{1}{4\eps^2} + \frac{c}{2\eps}$ or $| N(\ell_i) \cap T| \le \frac{1}{4\eps^2} - \frac{c}{2\eps}$.

Let $S = \{\ell_i\} \cup (R \setminus T)$. The cut $w(S, V \setminus S)$ consists of forward edges from $\ell_i$ to $T$ and backward edges from $(R \setminus T)$ to $(L \setminus \{\ell_i\})$. Ideally, if Bob knows $w(S, V \setminus S)$, he can subtract the weight of backward edges to obtain $w(\ell_i, T) = \frac{1}{\eps^2} + | N(\ell_i) \cap T|$ and recover $| N(\ell_i) \cap T|$. However, Bob can only get a $(1 \pm \eps)$-approximation of $w(S, V \setminus S)$, which may have $\Theta(\beta/\eps^3)$ additive error because $w(S, V \setminus S) = \Theta(\beta/\eps^4)$. With this much error, Bob cannot distinguish between the two cases.

To overcome this issue, we follow the idea of~\cite{ACK+15}. Intuitively, when $c$ is small, roughly half of $\ell_i \in L$ satisfy $|N(\ell_i) \cap T| \ge \frac{1}{4\eps^2} + \frac{c}{2\eps}$.
By enumerating all subsets $Q \subset L$ of size $|Q| = \frac{|L|}{2}$, Bob can find a set $Q$ such that most nodes $v \in Q$ satisfy $|N(\ell_i) \cap T| \ge \frac{1}{4\eps^2} + \frac{c}{2\eps}$.
Since $|Q| = \frac{\beta}{2\eps^2}$, the $\frac{c}{2\eps}$ bias per node adds up to roughly $\frac{c \beta}{2\eps^3}$, which can be detected even with $\Theta(\beta/\eps^3)$ error.

To prove Lemma~\ref{lem:for_all_unit}, we need the following two technical lemmas, which are essentially proved in~\cite{ACK+15}.

\begin{lemma} [Claim 3.5 in~\cite{ACK+15}]
    \label{lem:for_all_1}
    Let $c > 0$ and $\frac{\sqrt{\beta}}{\eps} \ge \frac{10}{c}$.
    Consider the following sets:
    \begin{align*}
    L_{\mathrm{high}} &= \{\ell_i \in L: |N(\ell_i) \cap T| \ge \frac{1}{4\eps^2} + \frac{c}{2\eps}\} \; , \text{ and} \\
    L_{\mathrm{low}} &= \{\ell_i \in L: |N(\ell_i) \cap T| \le \frac{1}{4\eps^2} - \frac{c}{2\eps}\} \; .
    \end{align*}
    With probability at least $0.98$, we have $\frac{1}{2} - 10c \le \frac{|L_{\mathrm{high}}|}{|L|} \le \frac{1}{2}$ and $\frac{1}{2} - 10c \le \frac{|L_{\mathrm{low}}|}{|L|} \le \frac{1}{2}$. 
\end{lemma}

\begin{lemma}[Lemma 3.4 in~\cite{ACK+15}]
    \label{lem:for_all_2}
    Let $c_1 > 0$ be a sufficiently small universal constant.
    Suppose one can approximate $w(U, T)$ with additive error $c_1 \beta / \eps^3$ for every $U \subset L$ with $|U| = \frac{|L|}{2}$. Let $Q \subset L$ be the subset with the highest (approximate) cut value. Then, with probability at least $0.96$, we have $\frac{|L_{\mathrm{high}} \cap Q|}{|L_{\mathrm{high}}|} \ge \frac{4}{5}$.
\end{lemma}

We are now ready to prove Lemma~\ref{lem:for_all_unit}.

\begin{proof}[Proof of Lemma~\ref{lem:for_all_unit}]
We reduce from the distributional Gap-Hamming problem (Lemma~\ref{lem:gap_hamming}) with $h = \beta^2/\eps^2$.
We re-index Alice's $h$ strings as $s_{i, j}$, where $i \in [\beta / \eps^2]$ and $j \in [\beta]$.

We construct a directed bipartite graph $G$ with two parts $L$ and $R$, where $|L| = |R| = \beta / \eps^2$.
Let $L = \{\ell_1, \ell_2, \ldots, \ell_{\beta/\eps^2}\}$.
We partition $R$ into $\beta$ disjoint sets with $|R_1| = \ldots = |R_\beta| = 1/\eps^2$.
We encode $s_{i, j} \in \{0, 1\}^{1/\eps^2}$ using the edges from $\ell_i$ to $R_j$: For node $\ell_i$ and the $v$-th node in $R_j$, the {\em forward} edge $(\ell_i, v)$ has weight $s_{i, j}(v) + 1$, and the {\em backward} edge $(v, \ell_i)$ has weight $1/\beta$.

Note that $G$ is $(2\beta)$-balanced.
We will show that given a $(1 \pm c_2 \eps)$ for-all cut sketch of $G$ for some constant $c_2 > 0$, Bob can decide whether $\Delta(s_{i, j}, t) \ge \frac{1}{2\eps^2} + \frac{c}{\eps}$ or $\Delta(s_{i, j}, t) \le \frac{1}{2\eps^2} - \frac{c}{\eps}$ with probability at least $2/3$.
Consequently, by Lemma~\ref{lem:gap_hamming}, any for-all cut sketching algorithm must output $\Omega(h/\eps^2) = \Omega(\beta^2/\eps^4) = \Omega(\beta'^2/\eps'^4)$ bits for $\beta' = 2\beta$ and $\eps' = c_2 \eps$.

Bob enumerates every $U \subset L$ with $|U| = \frac{|L|}{2} = \frac{\beta}{2\eps^2}$ and uses the cut sketch to approximate $w(U, T)$, where $T \subset R_j$ corresponds to the positions of $1$ in Bob's string $t$ and $|T| = \frac{1}{2\eps^2}$.
Let $S = U \cup (R \setminus T)$.
The cut $(S, V \setminus S)$ has $\frac{\beta}{4\eps^4}$ forward edges from $U$ to $T$ with weights $1$ or $2$, and $\big(\frac{\beta}{\eps^2} - \frac{1}{2\eps^2}\big)\big(\frac{\beta}{2\eps^2}\big) = O\big(\frac{\beta^2}{\eps^4}\big)$ backward edges from $(R \setminus T)$ to $(L \setminus U)$ with weight $\frac{1}{\beta}$.
The total weight of these edges is $O(\beta / \eps^4)$.
Therefore, given a $(1 \pm c_2 \eps)$ for-all cut sketch, Bob can subtract the fixed weight of the backward edges and approximate $w(U, T)$ with additive error $O(c_2\beta / \eps^3)$.
When $c_2$ is sufficiently small, this additive error is at most $c_1 \beta/\eps^3$.
By Lemma~\ref{lem:for_all_2}, Bob can find $Q \subset L$ with $|Q| = \frac{|L|}{2}$ such that $\frac{|L_{\mathrm{high}} \cap Q|}{|L_{\mathrm{high}}|} \ge \frac{4}{5}$.
Finally, if $\ell_i \in Q$, Bob decides $|N(\ell_i) \cap T| \ge \frac{1}{4\eps^2} + \frac{c}{2\eps}$ and $\Delta(s_{i, j}, t) \le \frac{1}{2\eps^2} - \frac{c}{\eps}$; and if $\ell_i \notin Q$, Bob decides $\Delta(s_{i, j}, t) \ge \frac{1}{2\eps^2} + \frac{c}{\eps}$.

Suppose Bob's index is $(i, j)$.
Notice that Bob uses $j$ to determine which $R_j$ to look at, but does not use any information about $i$.
Therefore, when $\Delta(s_{i, j}, t) \le \frac{1}{2\eps^2} - \frac{c}{\eps}$ and $|N(\ell_i) \cap T| \ge \frac{1}{4\eps^2} + \frac{c}{2\eps}$,
\[
\mathbf{Pr} [i \in Q] = \frac{|L_{\mathrm{high}} \cap Q|}{|L_{\mathrm{high}}|} \ge \frac{4}{5} \; .
\]
Conversely, when $\Delta(s_{i, j}, t) \le \frac{1}{2\eps^2} + \frac{c}{\eps}$ and $|N(\ell_i) \cap T| \le \frac{1}{4\eps^2} - \frac{c}{2\eps}$, because $L_{\mathrm{low}} \cap L_{\mathrm{high}} = \varnothing$,
\[
\mathbf{Pr} [i \notin Q] = \frac{|L_{\mathrm{low}} \setminus Q|}{|L_{\mathrm{low}}|} = \frac{|L_{\mathrm{low}}| - |L_{\mathrm{low}} \cap Q|}{|L_{\mathrm{low}}|} \ge \frac{|L_{\mathrm{low}}| - \frac{1}{5}|L_{\mathrm{high}}|}{|L_{\mathrm{low}}|} \ge \frac{3}{4} \; .
\]
The last inequality holds because $|L_{\mathrm{low}}| \ge 0.4 |L|$ and $|L_{\mathrm{high}}| \le 0.5 |L|$ by Lemma~\ref{lem:for_all_1} when $c \le 0.1$.

We analyze Bob's success probability.
Lemma~\ref{lem:for_all_1} fails with probability at most $0.02$, Lemma~\ref{lem:for_all_2} fails with probability at most $0.04$, and the for-all cut sketch fails with probability at most $0.01$.~\footnote{The probability that a for-all cut sketch (Definition~\ref{def:for-all}) preserves all cuts simultaneously can be boosted from $2/3$ to $99/100$, e.g., by running the sketching and recovering algorithms $O(1)$ times and taking the median. This increases the length of Alice's message by a constant factor, which does not affect our asymptotic lower bound.}
If they all succeed, Bob's probability of answering correctly is at least $\min\left(\frac{|L_{\mathrm{high}} \cap Q|}{|L_{\mathrm{high}}|}, \frac{|L_{\mathrm{low}} \setminus Q|}{|L_{\mathrm{low}}|}\right) \ge \frac{3}{4}$.
Bob's overall fail probability is at most $0.02 + 0.04 + 0.01 + 0.25 < 1/3$.
\end{proof}

We next consider the case with general values of $n, \beta$, and $\eps$, and prove Theorem~\ref{thm:for_all}.

\begin{proof} [Proof of Theorem~\ref{thm:for_all}]
    Let $k = \beta / \eps^2$. We assume w.l.o.g. that $k$ is an integer and $n$ is a multiple of $k$. We reduce from the distributional Gap-Hamming problem in Lemma~\ref{lem:gap_hamming} with $h = \Omega(n\beta)$.
    We will show that Alice's strings can be encoded into a graph $G$ such that

    {\em (i)} $G$ has $n$ nodes and is $(2\beta)$-balanced, and 
    
    {\em (ii)} After receiving a string $t$, an index $q \in [h]$, and a $(1 \pm c_2 \eps)$ for-all cut sketch of $G$ for some universal constant $c_2 > 0$, Bob can distinguish whether $\Delta(s_q, t) \le \frac{1}{2\eps^2} - \frac{c}{\eps}$ or $\Delta(s_q, t) \ge \frac{1}{2\eps^2} + \frac{c}{\eps}$ with probability at least $2/3$. 
    
    Consequently, by Lemma~\ref{lem:gap_hamming}, any for-all cut sketching algorithm must output $\Omega(h /\eps^2) = \Omega(n \beta / \eps^2) = \Omega(n \beta' / \eps'^2)$ bits for $\beta' = 2\beta$ and $\eps' = c_2 \eps$.

    We first describe the construction of $G$. We partition the $n$ nodes into $\ell = n / k \ge 2$ disjoint sets $V_1, V_2, \cdots, V_\ell$, each containing $k$ nodes.
    Let $s_1, s_2, \cdots, s_h \in \{0, 1\}^{1/\eps^2}$ be Alice's random strings where $h = (t - 1)(\beta^2/\eps^2) = \Omega((n/k)(\beta^2 / \eps^2)) = \Omega(n\beta)$. We partition the $h$ strings into $(t - 1)$ disjoint sets $S_1, S_2, \cdots, S_{t - 1}$, each with $(\beta^2 / \eps^2)$ strings. We then follow the same procedure as in Lemma~\ref{lem:for_all_unit} to encode $S_i$ into a complete bipartite graph between $V_i$ and $V_{i + 1}$. Notice that $S_i$ has $(\beta^2/\eps^2)$ strings and $|V_i| = |V_{i + 1}| = k = \beta/\eps^2$, which is the same setting as in Lemma~\ref{lem:for_all_unit}. 

    We can verify that $G$ is $(2\beta)$-balanced.
    This is because every edge $e$ has a reverse edge whose weight is at most $2 \beta$ times the weight of $e$.
    For every $u \in V_i$ and $v \in V_{i+1}$, the edge $(u, v)$ has weight $1$ or $2$, while the edge $(v, u)$ has weight $1/\beta$.

    We next show how Bob can distinguish between the two cases. Suppose Bob's index $q$ specifies a string encoded by the subgraph between $V_{i}$ and $V_{i + 1}$. Similar to the proof of Lemma~\ref{lem:for_all_unit}, we only need to show that given a $(1 \pm c_2\eps)$ for-all cut sketch, Bob can approximate $w(U, T)$ with additive error $O(\beta/\eps^3)$ for every $U \subset V_{i}$ with $|U| = \frac{|V_i|}{2} = \frac{\beta}{2\eps^2}$ and for some $T \subset V_{i+1}$ with $|T| = \frac{1}{2\eps^2}$. To see this, consider $S = U \cup \left(V_{i + 1} \setminus T\right) \bigcup_{j = i + 2}^{t} V_j$. The edges from $S$ to $(V \setminus S)$ are
\begin{itemize}[leftmargin=*]
    \item $\frac{\beta}{4\eps^4}$ forward edges from $U$ to $T$, each with weight $1$ or $2$.
    \item $\big(\frac{\beta}{\eps^2} - \frac{1}{2\eps^2}\big)\big(\frac{\beta}{2\eps^2}\big)$ backward edges from $(V_{i + 1} \setminus T)$ to $(V_{i} \setminus U)$, each with weight $\frac{1}{\beta}$.
    \item $\frac{\beta^2}{2\eps^4}$ backward edges from $U$ to $V_{i - 1}$ when $i \ge 2$, each with weight $\frac{1}{\beta}$.
\end{itemize}
    The total weight of these edges is $w(S, V \setminus S) = O(\beta / \eps^4)$. Consequently, given a $(1 \pm c_2\eps)$ cut sketch, Bob can subtract the fixed weight of the backward edges and approximate $w(U, T)$ with $O(c_2 \eps (\beta / \eps^4)) = O(c_2 \beta / \eps^3)$ additive error. Similar to the proof of Lemma~\ref{lem:for_all_unit}, for sufficiently small constant $c_2 > 0$, this will allow Bob to distinguish between the two cases $\Delta(s_q, t) \le \frac{1}{2\eps^2} - \frac{c}{\eps}$ or $\Delta(s_q, t) \ge \frac{1}{2\eps^2} + \frac{c}{\eps}$ with probability at least $2/3$.
\end{proof} 
\section{Local Query Complexity of Min-Cut}
\label{sec:lower-bound}

In this section, we present an $\Omega\big(\min\{m, \frac{m}{\eps^2k}\}\big)$ lower bound on the query complexity of approximating the global minimum cut of an undirected graph $G$ to a $(1 \pm \eps)$ factor in the local query model. Formally, we have the following theorem.

\MinCutLocalQuery*

To achieve this, we define a variant of the \textsc{2-SUM} communication problem in Section \ref{2sumprelims}, show a graph construction in Section \ref{graph}, and show that approximating \textsc{2-SUM} can be reduced to the minimum cut problem using our graph construction in Section \ref{reduction}.
In Section~\ref{sec:query_ub}, we will show that our lower bound is tight up to logarithmic factors.

\subsection{\textsc{2-SUM} Preliminaries} \label{2sumprelims}

Building off of the work of \cite{2sum}, we define the following variant of the $\textsc{2-SUM}(t, L, \alpha)$ problem.

\begin{definition}
    For binary strings $x = (x_1,\dots, x_L) \in \{0, 1\}^L$ and $y = (y_1, \dots, y_L) \in \{0, 1\}^L$, let $\textsc{INT}(x, y) = \sum_{i=1}^L x_i \wedge y_i$ denote the number of indices where $x$ and $y$ are both $1$.
    Let $\textsc{DISJ}(x, y)$ denote whether $x$ and $y$ are disjoint. That is, $\textsc{DISJ}(x, y) = 1$ if $\textsc{INT}(x, y) = 0$, and $\textsc{DISJ}(x, y) = 0$ if $\textsc{INT}(x, y) \ge 1$.
\end{definition}

\begin{definition}
    Suppose Alice has $t$ binary strings $(X^1, \dots, X^t)$ where each string $X^i \in \{0, 1\}^L$ has length $L$ and likewise Bob has $t$ strings $(Y^1, \dots, Y^t)$ each of length $L$. $\textsc{INT}(X^i, Y^i)$ is guaranteed to be either $0$ or $\alpha \geq 1$ for each pair of strings $(X^i, Y^i).$ Furthermore, at least $1/1000$ of the $(X^i, Y^i)$ pairs are guaranteed to satisfy $\textsc{INT}(X^i, Y^i) = \alpha.$ In the $\textsc{2-SUM}(t, L, \alpha)$ problem, Alice and Bob want to approximate $\sum_{i \in [t]} \textsc{DISJ}(X^i, Y^i)$ up to additive error $\sqrt{t}$ with high constant probability.
\end{definition}

\begin{lemma}
 \label{thm:basic2sum}
    To solve $\textsc{2-SUM}(t, L, 1)$ with high constant probability,
    the expected number of bits Alice and Bob need to communicate is $\Omega(tL)$.
\end{lemma}

\begin{proof}
    \cite{2sum} proved an expected communication complexity of $\Omega(tL)$ for $\textsc{2-SUM}(t, L, 1)$ without the promise that at least a $1/1000$ fraction of the $t$ string pairs intersect. Adding this promise does not change the communication complexity, because if $(X^1, \dots, X^t)$ and $(Y^1, \dots, Y^t)$ do not satisfy the promise, we can add a number of new $X^i$ and $Y^i$ to satisfy the promise and later subtract their contribution to approximate $\sum_{i \in [t]} \textsc{DISJ}(X^i, Y^i)$ with additive error $\Theta(\sqrt{t})$.
\end{proof}

\begin{theorem} \label{thm:ext2sum}
 To solve $\textsc{2-SUM}(t, L, \alpha)$ with high constant probability,
    the expected number of bits Alice and Bob need to communicate is $\Omega(tL/\alpha)$.    
\end{theorem}

\begin{proof}
    Consider an instance of $\textsc{2-SUM}(t, L/\alpha, 1)$ with Alice's strings $(X^1, \dots, X^t)$ and Bob's strings $(Y^1, \dots, Y^t)$ each with length $L/\alpha.$ For each of Alice's strings $X^i$ with length $L/\alpha$, we produce $X^{i,\alpha}$ (with length $L$) by concatenating $\alpha$ copies of $X^i$, and likewise we produce $Y^{i,\alpha}$ for each of Bob's strings $Y^i$. The setup where Alice has strings $(X^{1,\alpha}, \dots, X^{t, \alpha})$ and Bob has strings $(Y^{1,\alpha}, \dots, Y^{t,\alpha})$ is an instance of $\textsc{2-SUM}(t, L, \alpha)$. From Lemma~\ref{thm:basic2sum}, the communication complexity of $\textsc{2-SUM}(t, L/\alpha, 1)$ is $\Omega(tL/\alpha)$. Thus, the communication complexity of $\textsc{2-SUM}(t, L, \alpha)$ is $\Omega(tL/\alpha).$
\end{proof}

\subsection{Graph Construction} \label{graph}

Inspired by the graph construction from \cite{eden2017lower}, given two strings $x, y \in \{0, 1\}^N$, we construct a graph $G_{x,y}(V,E)$ such that $V$ is partitioned into $A$, $A'$, $B$ and $B'$, where $\card{A} = \card{A'} = \card{B} = \card{B'} = \sqrt{N} = \ell$. Note that since $\ell^2 = N$, we can index the bits in $x$ by $x_{i,j}$, where $1 \leq i, j \leq \ell$. We construct the edges $E$ according to the following rule:
\begin{align*}
    \begin{cases}
        (a_i, b_j'), (b_i, a_j') \in E & \text{if } x_{i,j} = y_{i,j} = 1 \\
        (a_i, a_j'), (b_i, b_j') \in E & \text{otherwise}
    \end{cases}
\end{align*}
\Cref{fig:graph-example} illustrates an example of the graph $G_{x,y}(V,E)$ when $x = 000000100$ and $y = 100010100$.

\begin{figure}[h]
    \centering
    \includegraphics[width=8cm]{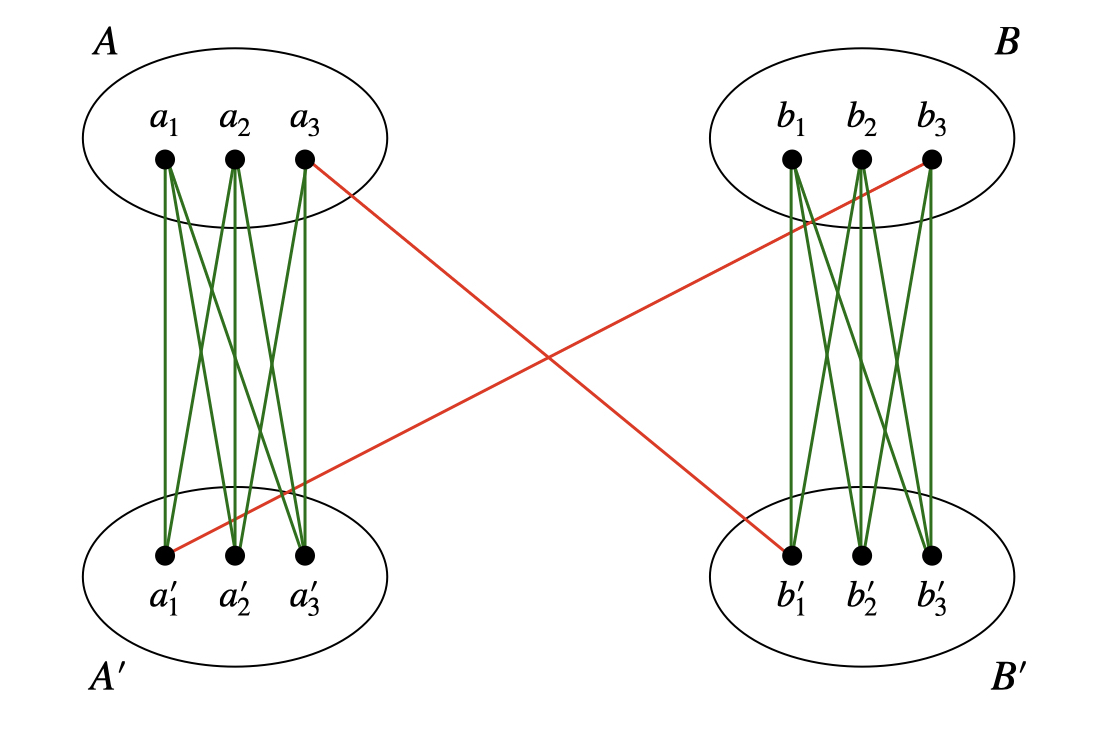}
    \caption{Example of $G_{x,y}(V,E)$ where $x = 000000100$ and $y = 100010100$. The red edges represent the intersection at $x_{31} = y_{31} = 1$. The green edges represent all the non-intersections in $x$ and $y$.}
    \label{fig:graph-example}
\end{figure}

  We will show that under certain assumptions about $N$ and $\mathrm{INT}(x,y)$, the number of intersections in $x,y$ is twice the size of the minimum cut in $G_{x,y}$.

\begin{lemma}
    \label{thm:k-connected}
    Given $x,y \in \{0,1\}^N$, if $\sqrt{N} \geq 3 \cdot \mathrm{INT}(x,y)$, then $\mathrm{MINCUT}(G_{x,y}) = 2 \cdot \mathrm{INT}(x,y)$. 
\end{lemma}

\begin{proof}
    
To prove this, we use some properties about $\gamma$-connectivity of a graph. A graph is \textit{$\gamma$-connected} if at least $\gamma$ edges must be removed from $G$ to disconnect it. In other words, if a graph $G$ is $\gamma$-connected, then $\mathrm{MINCUT}(G) \geq \gamma$. Equivalently, a graph $G$ is $\gamma$-connected if for every $u, v \in V$, there are at least $\gamma$ edge-disjoint paths between $u$ and $v$. Therefore, given $\mathrm{INT}(x,y) = \gamma$, if we can show that $G_{x,y}$ is $2 \gamma$-connected and there exists one cut of size exactly $2\gamma$, then we can show $\mathrm{MINCUT}(G_{x,y}) \ge 2 \gamma$. 
By the construction of the graph, it is easy to see that $\mathrm{CUT}(A \cup A', B \cup B')$ has size $2 \gamma$, since each intersection of $x,y$ produces two crossing edges in between. Therefore, all we need to show here is that if $\sqrt{N} \geq 3 \cdot \gamma$, then $G_{x,y}$ is $2 \gamma$-connected.

Similar to \cite{eden2017lower}, we prove this by looking at each pair of $u,v \in V$. Our goal is to show that for every $u,v \in V$, there exist at least $2\gamma$ edge-disjoint paths from $u$ to $v$.
\setcounter{case}{0} 
\begin{case} 
    \label{case:case1}
    $u, v \in A$ (or symmetrically $u,v \in A', B, B'$).
    For each pair $u, v \in A$, we have that there are at least $\ell - \gamma$ distinct common neighbors in $A'$. This is because one intersection at $x_{ij}$ and $y_{ij}$ implies that the edge $(a_i, a_j')$ is not contained in $E$, and would remove at most one common neighbor in $A'$. Since $\ell = \sqrt{N} \geq 3 \gamma$, we have that there are at least $\ell - \gamma \geq 2 \gamma$ distinct common neighbors in $A'$, which we denote by $u^{A'}_{1}, u^{A'}_{2}, \dots, u^{A'}_{2\gamma}$. Therefore, each path $u \rightarrow u^{A'}_i \rightarrow v$ is edge-disjoint, and we have at least $2\gamma$ edge-disjoint paths from $u$ to $v$, as shown in \Cref{fig:case1}.
    \begin{figure}[h]
        \centering
        \includegraphics[width=0.7\linewidth]{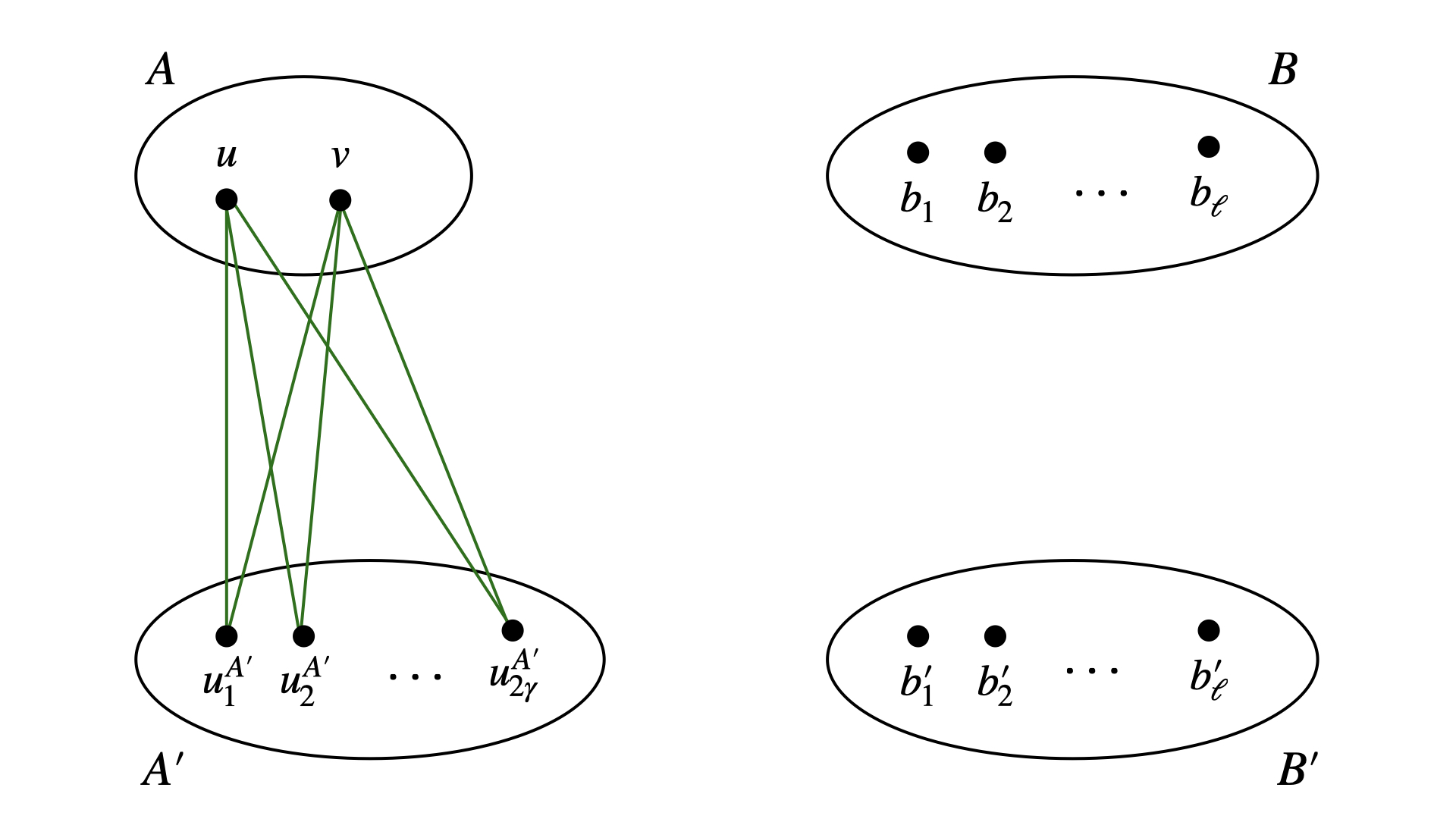}
        \caption{$u, v \in A$. We omit all the $(a_i, b_j')$, $(b_i, b_j')$, and $(b_i, a_j')$ edges.}
        \label{fig:case1}
    \end{figure}
\end{case}

\begin{case}
\label{case:case2}
    $u \in A, v \in A'$ (or symmetrically $u \in B, v \in B' $).
    Since $\ell - \gamma \geq 2 \gamma$, we have that $v$ has at least $2\gamma$ distinct neighbors in $A$, which we denote by $u^{A}_1, u^{A}_2, \dots, u^{A}_{2\gamma}$. From Case~\ref{case:case1}, we also have that each $u^{A}_i$ has at least $2 \gamma$ distinct common neighbors in $A'$. Therefore, we can choose $v^{A'}_1, v^{A'}_2, \dots, v^{A'}_{2\gamma}$ such that each path $u \rightarrow v^{A'}_i \rightarrow u^{A}_i \rightarrow v$ is edge-disjoint, so we have at least $2\gamma$ edge-disjoint paths from $u$ to $v$, as shown in \Cref{fig:case2}. Note that it may be the case where $u^{A}_i = u$. In this case, we can simply take the edge $(u, v)$ to be one of the edge-disjoint paths.

    \begin{figure}[h]
        \centering
        \includegraphics[width=0.7\linewidth]{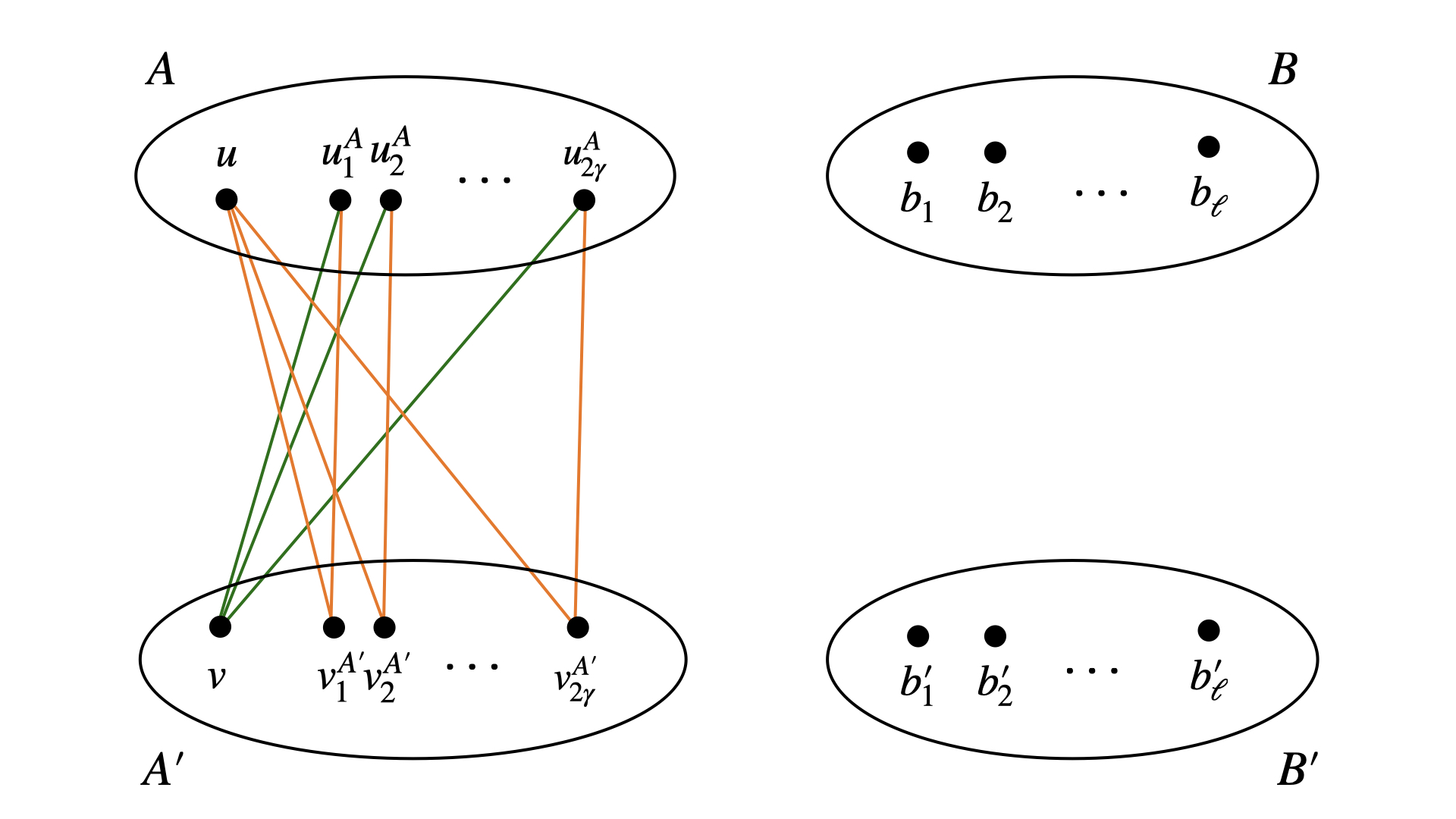}
        \caption{$u \in A, v \in A'$. We omit all the $(a_i, b_j')$, $(b_i, b_j')$, and $(b_i, a_j')$ edges. The green edges exist since $v$ has at least $2\gamma$ neighbors in $A$. The orange edges exist since $u^{A}_i$ and $u$ have at least $2\gamma$ common neighbors in $A'$. }
        \label{fig:case2}
    \end{figure}
\end{case}

\begin{case}
    \label{case:case3}
    $u \in A, v \in B'$ (or symmetrically $u \in A', v \in B$). 
    In this case, we show two sets of edge-disjoint paths, where each set has at least $\gamma$ edge-disjoint paths from $u$ to $v$, and the two sets of paths do not overlap. Overall, we have at least $2 \gamma$ edge-disjoint paths.

    The first set of paths $S_1$ uses the edges between $A'$ and $B$. Let $(w_1, x_1), (w_2, x_2), \dots, (w_{\gamma}, x_{\gamma}) \in A' \times B$ be the edges between $A'$ and $B$. Each of these edges represents one intersection in $x$ and $y$. Therefore, there are exactly $\gamma$ of them. From Case~\ref{case:case2}, we have that for every $w_i$, there are $2\gamma$ edge-disjoint paths from $u$ to $w_i$. Hence, for every $w_i$, we can choose a path from $u$ to $w_i$ and these $\gamma$ paths are edge-disjoint.
    \Cref{fig:case3.1} illustrates the paths $u \rightarrow u_i \rightarrow u'_i \rightarrow w_i \rightarrow x_i$. By symmetry, we can extend the paths from $x_i$ to $v$. This gives us $\gamma$ edge-disjoint paths from $u$ to $v$.

    \begin{figure}[h]
        \centering
        \includegraphics[width=0.7\linewidth]{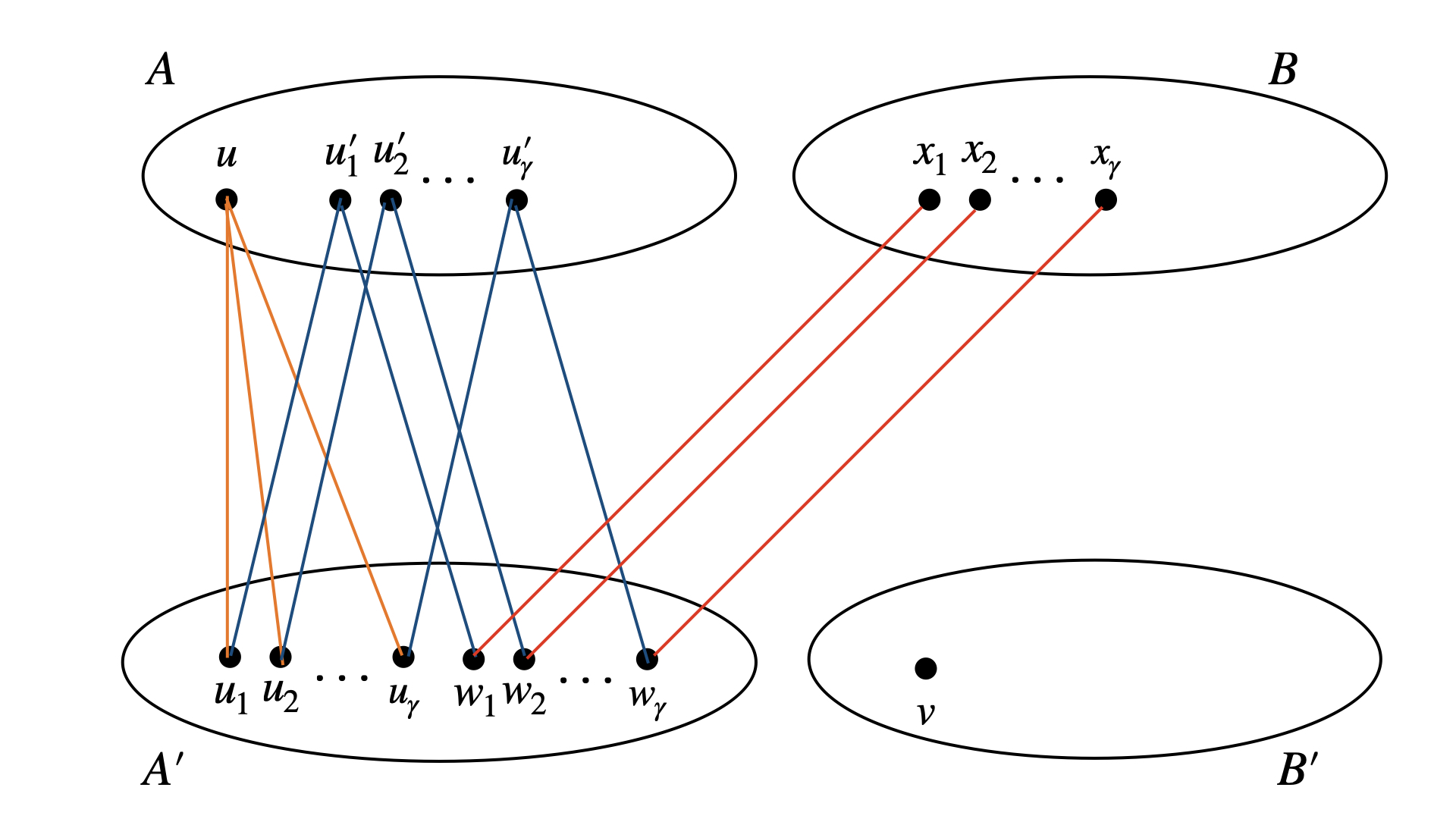}
        \caption{$u \in A, v \in B'$. The first set of paths $S_1$ goes from $u \rightarrow u_i \rightarrow u'_i \rightarrow w_i \rightarrow x_i$. We omit the paths from $x_i$ to $v$, as they are symmetric to the paths from $w_i$ to $u$. Once we extend the paths from $x_i$ to $v$, we have $\gamma$ edge-disjoint paths from $u$ to $v$. Note that the $w_i$ and $x_i$ may not be distinct.}
        \label{fig:case3.1}
    \end{figure}

    We now consider the second set of paths $S_2$. Let $$(y_1, z_1), (y_2, z_2), \dots, (y_{\gamma}, z_{\gamma}) \in A \times B'$$ be the distinct edges between $A$ and $B'$. Once again, it suffices to prove that there are $2\gamma$ edge-disjoint paths from $u$ to $y_i$, since the paths between $v$ to $z_i$ would be symmetric. From Case~\ref{case:case1}, we have that for every $y_i$, there are at least $2\gamma$ common neighbors between $y_i$ and $u$. Therefore, we can always find distinct $u''_1, u''_2, \dots, u''_{\gamma}$ such that the paths $u \rightarrow u''_i \rightarrow y_i$ are edge-disjoint, as shown in \Cref{fig:case3.2}. Once we extend the paths from $z_i$ to $v$, we have $\gamma$-edge disjoint paths in the second set.

    \begin{figure}[h]
        \centering
        \includegraphics[width=0.7\linewidth]{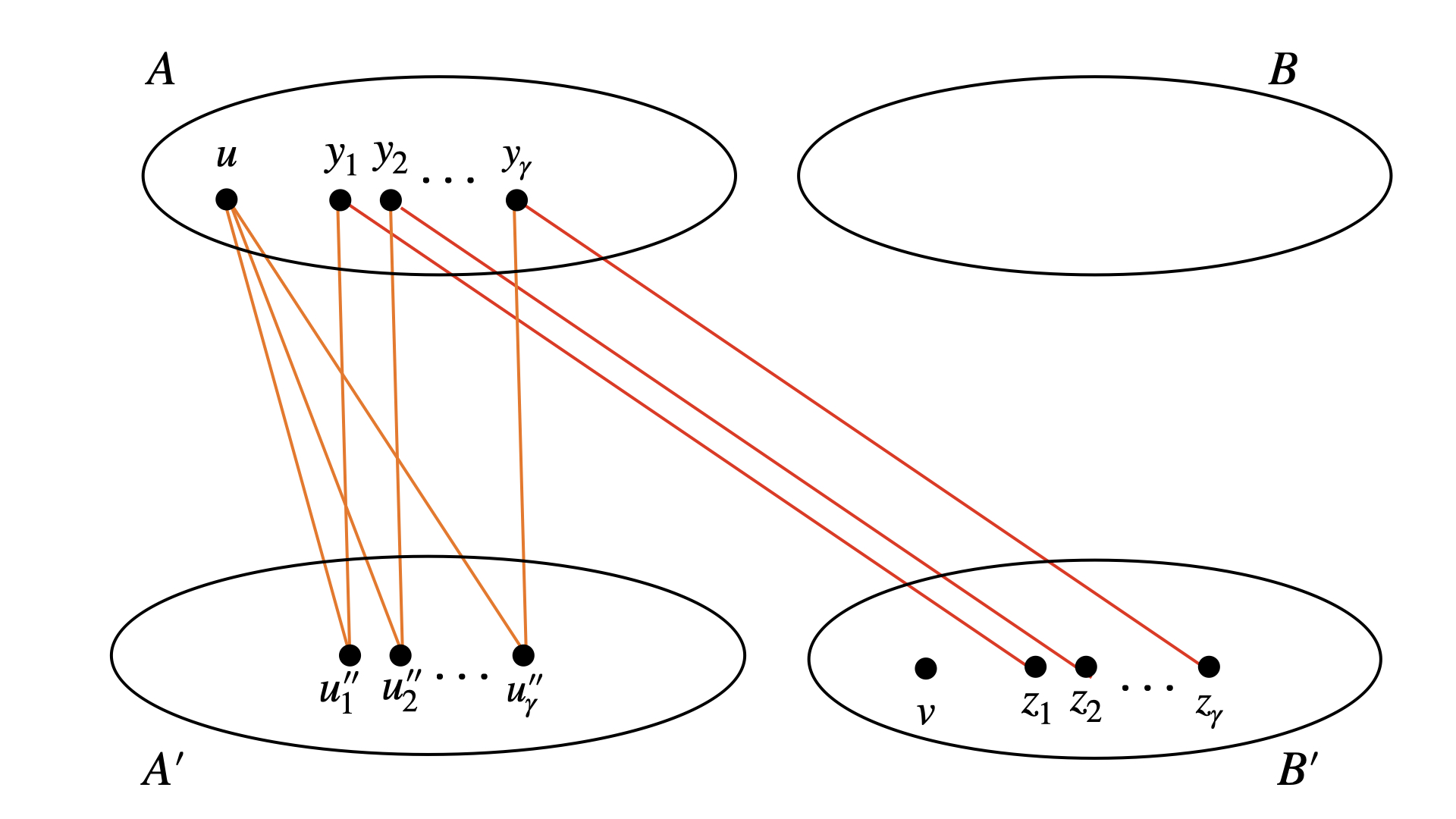}
        \caption{$u \in A, v \in B'$. The second set of paths $S_2$ goes from $u \rightarrow u''_i \rightarrow y_i \rightarrow z_i$. We omit the paths from $z_i$ to $v$, as they are symmetric to the paths from $y_i$ to $u$. Once we extend the paths from $x_i$ to $v$, we have $\gamma$ edge-disjoint paths from $u$ to $v$. Note that the $y_i$ and $z_i$ may not be distinct.}
        \label{fig:case3.2}
    \end{figure}

    Now we have two sets of paths $S_1$ and $S_2$, where both sets have at least $\gamma$ edge-disjoint paths. It remains to show that the paths in $S_1$ and $S_2$ can be edge-disjoint. Observe that the only possible edge overlaps between the paths from $u$ to the $w_i$ and paths from $u$ to the $y_i$ are $u \rightarrow u''_i$ and $u \rightarrow u_i$, since they are both neighbors of $u$. However, note that what we have shown is that for every $w_i$ or $y_i$, there are at least $2\gamma$ edge-disjoint paths from $u$ to $w_i$ or $y_i$. Therefore, one can choose $2\gamma$ edge-disjoint paths from $u$ to $w_i$ and $y_i$ such that $u_i'$ and $u_i''$ do not overlap. And similarly one can choose $2\gamma$ edge-disjoint paths from $v$ to the $z_i$ and the $x_i$. Overall, we have $2 \gamma$ edge-disjoint paths from $u$ to $v$.
\end{case}
    
\begin{case}
    $u \in A, v \in B$ (or symmetrically $u \in A', v \in B'$).
    This case is similar to Case~\ref{case:case3}, where we have two edge-disjoint sets $S'_1$ and $S'_2$. Consider the set of paths $S'_1$, where we use the edges $$(w_1, x_1), (w_2, x_2), \dots, (w_{\gamma}, x_{\gamma}) \in A' \times B.$$ We can construct the paths from $u$ to $w_i$ using the same way as for $S_1$ in Case~\ref{case:case3} (\Cref{fig:case3.1}). For the paths from $x_i$ to $v$, however, we construct them using the same way as in $S_2$ in Case~\ref{case:case3} (\Cref{fig:case3.2}). By connecting these paths, we obtain at least $\gamma$ edge-disjoint paths in $S'_1$. Similarly, we can also construct at least $\gamma$ edge-disjoint paths in $S'_2$, where we use the edges $$(y_1, z_1), (y_2, z_2), \dots, (y_{\gamma}, z_{\gamma}) \in A \times B'.$$ We follow the same way of choosing the paths in $S'_1$ and $S'_2$ that are edge-disjoint. \qedhere 
\end{case} 
\end{proof}

\subsection{Reducing 2-SUM to MINCUT}
\label{reduction}
In this section, we use the graph constructions in Section \ref{graph} to reduce the $\textsc{2-SUM}(t, L, \alpha)$ problem to $\textsc{MINCUT}$ and derive a lower bound on the number of queries in the local query model.

\begin{lemma} \label{lem:reduce}
    Given $M, \lambda > 0$, and $0 < \eps < 1$, suppose that we have any algorithm $\calA$ that can estimate the size of the minimum cut of a graph up to a $(1 \pm \eps)$ multiplicative factor with $T$ expected queries in the local query model. Then there exists an algorithm $\calB$ that can approximate $\textsc{2-SUM}(\eps^{-2}, \eps^2 M, \max\{\eps^2 \lambda, 1\})$ up to an additive error $\sqrt{\eps^{-2}} = \eps^{-1}$ using at most $O(T)$ bits of communication in expectation given $\sqrt{M} \geq 3 \max\{\lambda, \eps^{-2}\}$.
\end{lemma}

\begin{proof}
    We will show that the following algorithm $\calB$ satisfies the above conditions:
    \begin{enumerate}
        \item Given Alice's strings $(X^1, \dots, X^{\eps^{-2}})$ each of length $\eps^2 M$, let $x$ be the concatenation of Alice's strings having total length $\eps^{-2} (\eps^2 M) = M$. Similarly let $y \in \{0, 1\}^M$ be the concatenation of Bob's strings.

        \item Construct a graph $G_{x,y}$ as in Section \ref{graph} using the above concatenated strings as $x, y.$

        \item Run $\calA(G_{x,y})$ and output $\big(\frac{1}{\eps^2} - \frac{\calA(G_{x,y})}{2\max\{\eps^2 \lambda, 1\}} \big)$ as the solution to $\textsc{2-SUM}(\eps^{-2}, \eps^2 M, \max\{\eps^2 \lambda, 1\}).$
    \end{enumerate}

    For the \textsc{2-SUM} problem, let $r = \eps^{-2} - \sum_{i \in [\eps^{-2}]} \textsc{DISJ}(X^i, Y^i)$ be the number of string pairs with intersections. Since there are $\eps^{-2}$ pairs $(X^i, Y^i)$, $r$ is at most $\eps^{-2}$. From our definition of \textsc{2-SUM}, each intersecting string pair has $\max\{\eps^2 \lambda, 1\}$ intersections. $x, y$ are formed by concatenations, so $\textsc{INT}(x, y) = r \max\{\eps^2 \lambda, 1\}$.  Since $\sqrt{M} \geq 3 \max\{\lambda, \eps^{-2}\} = 3\eps^{-2} \max\{\eps^2 \lambda, 1\} \geq 3r \max\{\eps^2 \lambda, 1\} = 3 \cdot \textsc{INT}(x, y)$, Lemma \ref{thm:k-connected} is applicable to $G_{x, y}$ so that $$\textsc{MINCUT}(G_{x,y}) = 2r \max\{\eps^2 \lambda, 1\}.$$
    Since $\calA$ approximates $\textsc{MINCUT}$ up to a $(1 \pm \eps)$ factor, $\calA(G_{x,y}) = 2r (1 \pm \eps) \max\{\eps^2 \lambda, 1\}$. Thus, $\calB$'s output to the $\textsc{2-SUM}$ problem is within $(\eps^{-2} - r) \pm r\eps = \sum_{i \in [\eps^{-2}]} \textsc{DISJ}(X^i, Y^i) \pm r\eps$. Recall that $r \leq \eps^{-2}$. We can see that $\calB$ approximates $\textsc{2-SUM}(\eps^{-2}, \eps^2 M, \max\{\eps^2 \lambda, 1\})$ up to additive error $\eps^{-1}$.
    
    To compare the complexities of $\calA$ and $\calB$, recall $\calA$ is measured by degree, neighbor, and pair queries,  whereas $\calB$ is measured by bits of communication. Given the construction of $G_{x,y}$, as shown in~\cite{eden2017lower}, degree, neighbor, and pair queries can each be simulated using at most $2$ bits of communication:

    \begin{itemize}[leftmargin=*]
        \item Degree queries: each vertex in $G_{x,y}$ has degree $\sqrt{M}$ so Alice and Bob do not need to communicate to simulate degree queries.

        \item Neighbor queries: assuming an ordering where $a_i$'s $j$'th neighbor is either $a'_j$ or $b'_j$, Alice and Bob can exchange $x_{i,j}$ and $y_{i,j}$ with $2$ bits of communication to simulate a neighbor query.

        \item Pair queries: Alice and Bob can exchange $x_{i,j}$ and $y_{i,j}$ with $2$ bits of communication to determine whether edges $(a_i, b'_j)$ and $(b_i, a'_j)$ exist.
    \end{itemize}
    As each of $\calA$'s queries can be simulated using up to $2$ bits of communication in $\calB$, $\calB$ can use $O(T)$ bits of communication to simulate $T$ queries in $\calA$. So we have established a reduction from approximating $\textsc{2-SUM}(\eps^{-2}, \eps^2 M, \max\{\eps^2 \lambda, 1\})$ up to additive error $\eps^{-1}$ to approximating $\textsc{MINCUT}$ up to a $(1 \pm \eps)$ multiplicative factor.
\end{proof}

We are now ready to prove Theorem~\ref{thm:min_cut}.

\begin{proof}[Proof of Theorem~\ref{thm:min_cut}]
    Given an instance of $\textsc{2-SUM}(\eps^{-2}, \eps^2 m,$ $\max\{\eps^2 k, 1\})$, consider the same way of constructing the graph $G_{x, y}$ in Lemma~\ref{lem:reduce}. From the construction of $G_{x,y}$, the number of edges is $2m$ since each of pair $(x_i, y_i)$ corresponds to $2$ edges. Using the promise from $\textsc{2-SUM}$, we get that $r \ge \eps^{-2}/1000$, where $r = \sum_{i \in [\eps^{-2}]} \textsc{DISJ}(X^i, Y^i)$, which means that the size of the minimum cut of $G_{x, y}$ is $2r\cdot\max\{\eps^{2}k, 1\} \ge \Omega( \max\{k,\eps^{-2}\})$. When $k \ge \eps^{-2}$, we have that the size of the minimum cut of $G_{x,y}$ is $\Omega(k)$, and from Lemma~\ref{lem:reduce} we obtain that any algorithm $\mathcal{A}$ that satisfies the guarantee on the distribution of $G_{x,y}$ must have $\Omega(m/(\eps^{2}k))$ queries in expectation. When $k < \eps^{-2}$, the size of the minimum cut of $G_{x,y}$ is $\Omega(\eps^{-2})$ and similarly we get that any algorithm $\mathcal{A}$ that satisfies the guarantee on the distribution of $G_{x,y}$ must use $\Omega(m)$ queries in expectation. Combining the two, we finally obtain an $\Omega(\min\{m, \frac{m}{\eps^2k}\})$ lower bound on the expected number of queries in the local query model. 
\end{proof}

\subsection{Almost Matching Upper Bound}
\label{sec:query_ub}

In this section, we will show that our lower bound is tight up to logarithmic factors. In the work of~\cite{globalmincut}, the authors presented an algorithm that uses $O(\frac{m}{k}\cdot \poly(\log n, 1/\eps))$ queries, where $k$ is the size of the minimum cut. We will show that, despite their analysis giving a dependence of $1/\eps^4$, a slight modification of their algorithm yields a dependence of $1/\eps^2$. Formally, we have the following theorem. 

\begin{theorem}[essentially~\cite{globalmincut}]
\label{thm:cut_query_ub}
    There is an algorithm that solves the minimum cut query problem up to a $(1 \pm \eps)$-multiplicative factor with high constant probability in the local query model. Moreover, the expected number of queries used by this algorithm is $\tilde{O}\big(\frac{m}{\eps^2 k}\big)$.
\end{theorem}

To prove Theorem~\ref{thm:cut_query_ub}, we first give a high-level description of the algorithm in~\cite{globalmincut}. The algorithm is based on the following sub-routine.

\begin{lemma}[\cite{globalmincut}]
    There exists an algorithm $\textsc{Verify-Guess}(D, t, \eps)$ which makes $\tilde{O}(\eps^{-2} m / t)$ queries in expectation such that (here $D$ is the degree of each node)
    \begin{enumerate}[leftmargin=*]
        \item If $t \ge \frac{2000 \log n}{\eps^2} \cdot k$, then $\textsc{Verify-Guess}(D, t, \eps)$ rejects $t$ with probability at least $1 - \frac{1}{\poly(n)}$. 
        \item If $t \le  k$, then $\textsc{Verify-Guess}(D, t, \eps)$ accepts $t$ and outputs a $(1 \pm \eps)$-approximation of $k$ with probability at least $1 - \frac{1}{\poly(n)}$.
    \end{enumerate}
\end{lemma}

Given the above sub-routine, the algorithm initializes a guess $t = \frac{n}{2}$ for the value of the minimum cut $k$ and proceeds as follows:
\begin{itemize}
    \item if $\textsc{Verify-Guess}(D, t, \eps)$ rejects $t$, set $t = t / 2$ and repeat the process.
    \item if $\textsc{Verify-Guess}(D, t, \eps)$ accepts $t$, set $t = t/\kappa$ where $\kappa = \frac{2000\log n }{\eps^2}$.
    Let $\tilde{k} = \textsc{Verify-Guess}(D, t, \eps)$ and return the value of $\tilde{k}$ as the output.
\end{itemize}

To analyze the query complexity of the algorithm, notice that when $\textsc{Verify-Guess}$ first accepts $t$, we have that 
$\frac{k}{2} < t < \kappa k$.
which means that $t/\kappa < k$ and hence one call to $\textsc{Verify-Guess}(D, t/\kappa, \eps)$ will get the desired output. However, at a time in $t = \Theta(k/\kappa)$, the $\textsc{Verify-Guess}$ procedure needs to make $\tilde{O}\big(\frac{m}{\eps^{4}k}\big)$ queries in expectation.

To avoid this, the crucial observation is that, during the above binary search process, the error parameter of $\textsc{Verify-Guess}(D, t, \eps)$ does not have to be set to $\eps$. Using a small constant $\beta_0$ is sufficient. This way, when $\textsc{Verify-Guess}(D, t, \beta_0)$ first accepts $t$, we have 
$\frac{k}{2} < t < c \log(n) \cdot k$,
where $c$ is a constant. Consequently, the output of $\textsc{Verify-Guess}(D, t/(c \log n), \eps)$ will satisfy the error guarantee. Using the analysis in~\cite{globalmincut}, we can show that the query complexity of the new algorithm is $\tilde{O}(\frac{m}{\eps^{2}k})$.

\section*{Acknowledgement}
Yu Cheng is supported in part by NSF Award CCF-2307106. Honghao Lin and David Woodruff would like to thank support from the National Institute of Health (NIH) grant 5R01 HG 10798-2, and a Simons Investigator Award. Part of this work was done while D. Woodruff was visiting the Simons Institute for the Theory of Computing.

\bibliography{reference}

\newcommand{\etalchar}[1]{$^{#1}$}
\begin{thebibliography}{BGMP21}

\bibitem[ACK{\etalchar{+}}16]{ACK+15}
Alexandr Andoni, Jiecao Chen, Robert Krauthgamer, Bo~Qin, David~P. Woodruff, and Qin Zhang.
\newblock On sketching quadratic forms.
\newblock In {\em Proceedings of the 2016 {ACM} Conference on Innovations in Theoretical Computer Science (ITCS)}, pages 311--319, 2016.

\bibitem[AGM12]{AGM12}
Kook~Jin Ahn, Sudipto Guha, and Andrew McGregor.
\newblock Graph sketches: sparsification, spanners, and subgraphs.
\newblock In {\em Proceedings of the 31st {ACM} {SIGMOD-SIGACT-SIGART} Symposium on Principles of Database Systems (PODS)}, pages 5--14, 2012.

\bibitem[BGMP21]{globalmincut}
Arijit Bishnu, Arijit Ghosh, Gopinath Mishra, and Manaswi Paraashar.
\newblock Query complexity of global minimum cut.
\newblock In {\em Approximation, Randomization, and Combinatorial Optimization. Algorithms and Techniques (APPROX/RANDOM)}, volume 207 of {\em Leibniz International Proceedings in Informatics (LIPIcs)}, pages 6:1--6:15, 2021.

\bibitem[BK96]{BK96}
Andr{\'{a}}s~A. Bencz{\'{u}}r and David~R. Karger.
\newblock Approximating \emph{s-t} minimum cuts in \emph{{\~{O}}}(\emph{n}\({}^{\mbox{2}}\)) time.
\newblock In {\em Proceedings of the 28th Annual {ACM} Symposium on the Theory of Computing (STOC)}, pages 47--55, 1996.

\bibitem[BSS12]{BSS12}
Joshua~D. Batson, Daniel~A. Spielman, and Nikhil Srivastava.
\newblock Twice-{R}amanujan sparsifiers.
\newblock {\em {SIAM} J. Comput.}, 41(6):1704--1721, 2012.

\bibitem[CCPS21]{CCP+21}
Ruoxu Cen, Yu~Cheng, Debmalya Panigrahi, and Kevin Sun.
\newblock Sparsification of directed graphs via cut balance.
\newblock In {\em 48th International Colloquium on Automata, Languages, and Programming (ICALP)}, volume 198 of {\em LIPIcs}, pages 45:1--45:21, 2021.

\bibitem[CGP{\etalchar{+}}23]{ChuGPSSW23}
Timothy Chu, Yu~Gao, Richard Peng, Sushant Sachdeva, Saurabh Sawlani, and Junxing Wang.
\newblock Graph sparsification, spectral sketches, and faster resistance computation via short cycle decompositions.
\newblock {\em {SIAM} J. Comput.}, 52(6):S18--85, 2023.

\bibitem[CKK{\etalchar{+}}18]{CohenKKPPRS18}
Michael~B. Cohen, Jonathan~A. Kelner, Rasmus Kyng, John Peebles, Richard Peng, Anup~B. Rao, and Aaron Sidford.
\newblock Solving directed {L}aplacian systems in nearly-linear time through sparse {LU} factorizations.
\newblock In {\em Proceedings of the 59th {IEEE} Annual Symposium on Foundations of Computer Science (FOCS)}, pages 898--909, 2018.

\bibitem[CKST19]{CKST19}
Charles Carlson, Alexandra Kolla, Nikhil Srivastava, and Luca Trevisan.
\newblock Optimal lower bounds for sketching graph cuts.
\newblock In {\em Proceedings of the 30th Annual {ACM-SIAM} Symposium on Discrete Algorithms (SODA)}, pages 2565--2569, 2019.

\bibitem[EMPS16]{EMPS16}
Alina Ene, Gary~L. Miller, Jakub Pachocki, and Aaron Sidford.
\newblock Routing under balance.
\newblock In {\em Proceedings of the 48th Annual {ACM} {SIGACT} Symposium on Theory of Computing (STOC)}, pages 598--611, 2016.

\bibitem[ER18]{eden2017lower}
Talya Eden and Will Rosenbaum.
\newblock Lower bounds for approximating graph parameters via communication complexity.
\newblock In {\em Approximation, Randomization, and Combinatorial Optimization. Algorithms and Techniques (APPROX/RANDOM)}, volume 116 of {\em Leibniz International Proceedings in Informatics (LIPIcs)}, pages 11:1--11:18, 2018.

\bibitem[FHHP19]{FHHP11}
Wai~Shing Fung, Ramesh Hariharan, Nicholas J.~A. Harvey, and Debmalya Panigrahi.
\newblock A general framework for graph sparsification.
\newblock {\em {SIAM} J. Comput.}, 48(4):1196--1223, 2019.

\bibitem[IT18]{IT18}
Motoki Ikeda and Shin{-}ichi Tanigawa.
\newblock Cut sparsifiers for balanced digraphs.
\newblock In {\em Approximation and Online Algorithms - 16th International Workshop (WAOA)}, volume 11312 of {\em Lecture Notes in Computer Science}, pages 277--294, 2018.

\bibitem[JS18]{JambulapatiS18}
Arun Jambulapati and Aaron Sidford.
\newblock Efficient {$\tilde O(n/\epsilon)$} spectral sketches for the {L}aplacian and its pseudoinverse.
\newblock In {\em Proceedings of the 29th Annual {ACM-SIAM} Symposium on Discrete Algorithms (SODA)}, pages 2487--2503, 2018.

\bibitem[KNR01]{KNR01}
Ilan Kremer, Noam Nisan, and Dana Ron.
\newblock Errata for: ``on randomized one-round communication complexity''.
\newblock {\em Comput. Complex.}, 10(4):314--315, 2001.

\bibitem[KP12]{KP12}
Michael Kapralov and Rina Panigrahy.
\newblock Spectral sparsification via random spanners.
\newblock In {\em Innovations in Theoretical Computer Science (ITCS)}, pages 393--398, 2012.

\bibitem[LS17]{LeeS17}
Yin~Tat Lee and He~Sun.
\newblock An {SDP}-based algorithm for linear-sized spectral sparsification.
\newblock In {\em Proceedings of the 49th Annual {ACM} {SIGACT} Symposium on Theory of Computing (STOC)}, pages 678--687, 2017.

\bibitem[McG14]{mcgregor2014graph}
Andrew McGregor.
\newblock Graph stream algorithms: a survey.
\newblock {\em ACM SIGMOD Record}, 43(1):9--20, 2014.

\bibitem[RSW18]{RSW18}
Aviad Rubinstein, Tselil Schramm, and S.~Matthew Weinberg.
\newblock Computing exact minimum cuts without knowing the graph.
\newblock In {\em 9th Innovations in Theoretical Computer Science Conference (ITCS)}, volume~94 of {\em LIPIcs}, pages 39:1--39:16, 2018.

\bibitem[SS11]{SS11}
Daniel~A. Spielman and Nikhil Srivastava.
\newblock Graph sparsification by effective resistances.
\newblock {\em {SIAM} J. Comput.}, 40(6):1913--1926, 2011.

\bibitem[ST04]{ST04}
Daniel~A. Spielman and Shang{-}Hua Teng.
\newblock Nearly-linear time algorithms for graph partitioning, graph sparsification, and solving linear systems.
\newblock In {\em Proceedings of the 36th Annual {ACM} Symposium on Theory of Computing (STOC)}, pages 81--90, 2004.

\bibitem[ST11]{ST11}
Daniel~A. Spielman and Shang{-}Hua Teng.
\newblock Spectral sparsification of graphs.
\newblock {\em {SIAM} J. Comput.}, 40(4):981--1025, 2011.

\bibitem[SW19]{SaranurakW19}
Thatchaphol Saranurak and Di~Wang.
\newblock Expander decomposition and pruning: Faster, stronger, and simpler.
\newblock In {\em Proceedings of the 30th Annual {ACM-SIAM} Symposium on Discrete Algorithms (SODA)}, pages 2616--2635, 2019.

\bibitem[WZ14]{2sum}
David~P. Woodruff and Qin Zhang.
\newblock An optimal lower bound for distinct elements in the message passing model.
\newblock In {\em Proceedings of the 25th Annual ACM-SIAM Symposium on Discrete Algorithms (SODA)}, page 718–733, 2014.

\end{thebibliography}
\bibliographystyle{alpha}

\end{document}